\RequirePackage{amsmath}
\documentclass[a4paper, 10pt]{amsart}
\usepackage{amssymb, url, color}
\usepackage[nodayofweek]{datetime}
\usepackage{enumerate,stmaryrd} 
\usepackage{euscript,enumitem}
\usepackage[space, compress, sort]{cite}
\usepackage[overload]{empheq}
\usepackage[usenames,dvipsnames]{xcolor}
\usepackage[colorlinks=true, bookmarks=true, pdfstartview=FitH, pagebackref=true]{hyperref}
\usepackage{mathrsfs}
\usepackage{cancel}
\usepackage{tcolorbox}

\numberwithin{equation}{section}
\theoremstyle{plain}

\newtheorem{thm}{Theorem}[section]
\newtheorem{theorem}[thm]{Theorem}

\newtheorem{lemma}[thm]{Lemma}
\newtheorem{prop}[thm]{Proposition}
\newtheorem{proposition}[thm]{Proposition}

\newtheorem{claim}{Claim}

\theoremstyle{remark}

\theoremstyle{definition}

\newcounter{mnotecount}[section]

\renewcommand{\phi}{\varphi}

\newcommand{\Omegatil}{\widetilde{\Omega}}

\newcommand{\bR}{\mathbb{R}}
\newcommand{\bL}{\mathbb{L}}

\newcommand{\cE}{\mathcal{E}}

\newcommand{\cP}{\mathcal{P}}

\renewcommand{\hbar}{\overline{h}}

\newcommand{\nubar}{\overline{\nu}}

\newcommand{\gtil}{\widetilde{g}}

\def\into{\hookrightarrow}
\newcommand{\phibar}{\overline{\varphi}}

\newcommand{\etatil}{\widetilde{\eta}}

\newcommand{\ubar}{\overline{u}}

\newcommand{\Rb}{\mathbb{R}}
\newcommand{\dsp}[1]{\displaystyle{#1}}
\newcommand{\fraccheloue}{\frac{N}{2}+1}
\newcommand{\Scal}{\textrm{Scal}}
\newcommand{\Lb}{\mathbb{L}}
\newcommand{\tilg}{\widetilde{g}}

\newcommand{\Sring}{\mathring{S}}

\newcommand{\definedas}{\mathrel{\raise.095ex\hbox{\rm :}\mkern-5.2mu=}}

\let\<\langle 
\let\>\rangle

\DeclareMathOperator{\tr}{tr}
\DeclareMathOperator{\divg}{div}

\DeclareMathOperator{\vol}{Vol}

\DeclareMathOperator{\supp}{supp}

\newcommand{\ric}{\mathrm{Ric}}

\newcommand{\scal}{\mathrm{Scal}}

\def\XXint#1#2#3{{\setbox0=\hbox{$#1{#2#3}{\int}$}
\vcenter{\hbox{$#2#3$}}\kern-.5\wd0}}

\definecolor{lgreen} {RGB}{180,210,100}
\definecolor{dblue}  {RGB}{20,66,129}
\definecolor{ddblue} {RGB}{11,36,69}
\definecolor{lred}   {RGB}{220,0,0}
\definecolor{nred}   {RGB}{224,0,0}
\definecolor{norange}{RGB}{230,120,20}
\definecolor{nyellow}{RGB}{255,221,0}
\definecolor{ngreen} {RGB}{98,158,31}
\definecolor{dgreen} {RGB}{78,138,21}
\definecolor{nblue}  {RGB}{28,130,185}
\definecolor{jblue}  {RGB}{20,50,100}
\definecolor{Apricot}{RGB}{255, 170, 123}

\renewcommand{\phi}{\varphi}
\newcommand{\dd}{\,\mathrm{d}}
\newcommand{\tnu}{\widetilde{\nu}}
\newcommand{\tn}{\widetilde{n}}
\newcommand{\nablah}{~^{(\textbf{h})}\nabla}

\renewcommand{\leq}{\leqslant}
\renewcommand{\geq}{\geqslant}

\DeclareMathOperator{\DeltaL}{\Delta_{\bL}}

\newtcolorbox{mymathbox}[1][]{colback=white, sharp corners, #1}
\begin{document}

\title[Constraint equations on compact manifolds with boundaries]
{A class of solutions to the conformal constraint equations on compact 
manifolds 
with apparent horizon boundary conditions}

\begin{abstract}
This article is dedicated to solving the Einstein constraint equations with 
apparent horizon boundaries and freely specified mean curvature. The main novelty is that we study the conformal constraint equations assuming only low regularity.
\end{abstract}

\author[J.-D. Pailleron]{Jean-David Pailleron}
\address[J.-D. Pailleron]{
  Institut Denis Poisson\\
  Universit\'e de Tours\\
  Parc de Grandmont\\ 37200 Tours \\ France}
\email{jean-david.pailleron@idpoisson.fr}



\date{October 18, 2022}
\maketitle
\tableofcontents

\section{Introduction}

We consider the Cauchy problem in general relativity setting that a
$(n+1)$-manifold $\mathscr{M}$ and a globally hyperbolic Lorentzian metric $h$ 
describe the evolution of the gravitational field via Einstein's equation:
\begin{equation*}
 \underbrace{\ric^h_{\mu\nu}-\frac{\scal^h}{2} h_{\mu\nu}}_{G^h_{\mu\nu}}=\kappa 
T_{\mu\nu}.
\end{equation*}

Here $\ric^h$ and $\scal^h$ represent the Ricci tensor and the scalar curvature 
of $h$, $G^h$ is called the Einstein tensor, $T$ is the energy-momentum 
tensor modelling the matter and $\kappa$ is a constant equal to 
$\dsp{\frac{8\pi\mathcal{G}}{c^4}}$ (at least in dimension $n = 3$) where $\mathcal{G}$ is the gravitational constant and $c$ 
the speed of light in vacuum.\\

If $M$ is a spacelike hypersurface of $\mathscr{M}$ with unit normal 
$\etatil$, we can define on $M$:
\begin{itemize}
 \item the induced metric $\gtil=h_{|TM}$,
 \item the second fundamental form 
$\displaystyle{K:K(X,Y)=h(\nablah_{X}\etatil,Y)}$, where$\nablah$ represents 
the  Levi-Civita connection for the metric $h$.
\end{itemize}

Here, we will consider the case of vacuum; i.e. $T_{\mu\nu}=0$.
The objects $\gtil$ and $K$ are linked by the constraint equations, following 
from the Gauss-Codazzi equations (see e.g. \cite{BartIsen} for their derivation).

\begin{equation} \label{constraintequations}   
\begin{cases}
  \scal_{\gtil} -|K|^2_{\gtil} + (tr_{\gtil} K)^2&=0, \\
  \divg_{\gtil} K-\dd(tr_{\gtil} K)&=0.
\end{cases}
\end{equation}

It was proven by Y. Choquet-Bruhat and R. Geroch in \cite{CB1,CB2} (see 
also \cite{RingstromCauchyProblem} for a recent review of the subject) that if 
a $n$-manifold $M$, a Riemannian metric $\gtil$ on $M$ and a 2-tensor $K$ form 
a solution to the constraint equations, we can reconstruct the spacetime 
($\mathscr{M},h$) from the knowledge of $(M,\gtil,K)$.\\

The constraint equations (with or without boundary conditions) is 
underdetermined so the idea of the conformal method is to split the initial 
data $(M, \gtil ,K)$ into a given part and an unknown part that will have to be adjusted to satisfy \eqref{constraintequations}. We summarize the splitting that interests us.\\

\textbf{Conformal method (boundary-free):}
 \begin{itemize}[label={}]
  \item \textbf{Data:}
  \begin{itemize}[label=\textbullet]
  \item a Riemannian manifold $(M,g)$,
  \item a function $\tau : M\rightarrow \bR$,
  \item a symmetric 2-tensor $\sigma$ on $M$, traceless, transverse 
($\textrm{div}_g \sigma \equiv 0$). ($\sigma$ is then called a 
\emph{TT-tensor}).
 \end{itemize}
  \item \textbf{Unknowns:}
   \begin{itemize}[label=\textbullet]
    \item a vector field $W$,
    \item a positive function $\varphi$.
   \end{itemize}
\end{itemize}

What we do with the conformal method is to search for $\phi$ and $W$ so that,
\begin{eqnarray*}
\begin{cases}
  \gtil&=\phi^{N-2}g,\\
  K &=\displaystyle{\frac{\tau}{n}\gtil +\phi^{-2}(\sigma+\bL_{g} W).}
\end{cases}
\end{eqnarray*}

Here, we set $\dsp{N:=\frac{2n}{n-2}}$ and $\bL_g$ stands for the conformal Killing operator given by:
\begin{equation*}
 \bL_g W_{ij} :=\nabla_i W_j +\nabla_j W_i -\frac{2}{n}\nabla^k W_k g_{ij},
\end{equation*}
$\nabla$ being the Levi-Civita connection associated to $g$. We also note 
that $\tau=\tr_{\gtil} K$ is the mean curvature of the embedding $M \hookrightarrow \mathscr{M}$. Using this method, the system \ref{constraintequations} becomes:

\begin{subequations}\label{VCCEs}
\begin{empheq}[left=\empheqlbrace]{align}
-\frac{4(n-1)}{n-2}\Delta_g \phi + \scal_g \phi &= \frac{1-n}{n}\tau^2 
\phi^{N-1} +|\sigma+\bL_g W|^2_g \phi^{-N-1}, \label{eqLich}\\
\Delta_{\bL,g} W &= \frac{n-1}{n} \phi^N \dd\tau. \label{eqVect}
\end{empheq}
\end{subequations}

\vspace{0.5cm}
Here, $\Delta_g := \divg_g (\nabla_g \phi)$ is the standard Laplacian and $\Delta_{\bL,g} W :=\divg_g(\bL_g W)$ is commonly called the vector Laplacian. From now on, except when specified, all metric dependent objects will be expressed using $g$ as a reference metric except for those with a tilde that have $\tilg$ as a reference metric.
So we will omit most of the time the subscript $g$.
Of the two equations written above, the first one \eqref{eqLich} is called \emph{the Lichnerowicz equation} and the second one \eqref{eqVect},
\emph{the vector equation}.\\

The study of general relativity leads to consider singularities and the event horizon it (supposedly) lies within, an object commonly known as ``black hole''.
To adapt the problem of modelling black holes with constraint equations, we 
carve a ``hole'' in our hypersurface $M$ and thus add a boundary $\partial M$ 
that we will impose to be a \emph{marginally outer trapped surface} (MOTS) as 
we describe below. This hole that can potentially have multiple connected components if one wants 
to describe a universe containing several black holes will be the inside of 
the black hole and thus the solution to the constraint equations we will 
construct models the outside of the black hole(s).

We now introduce the notion of marginally outer trapped surface in general. These 
surfaces are indications defined using only the initial data that a spacetime 
contains black holes as it follows from the singularity theorem of Penrose and 
Hawking that the existence of a trapped surface leads to geodesic 
incompleteness of the spacetime (see e.g. \cite{Hawking}) In 
particular, we will loosely think of MOTS as describing the 
boundary of a black hole. The main interest of this notion lies in the fact that an event horizon 
asks for the complete knowledge of spacetime whereas an MOTS only depends on initial data.\\

We now give a description of MOTS, refering to \cite{wald} for more details.
Let $\Sigma$ be a two-sided hypersurface embedded in $M$, that will be chosen to be $\partial M$ later, so it defines two regions in $M$,
one of which we think as being the inside of the black hole. We 
denote by $\tnu$ the unit normal to $\Sigma$ in $M$ that points 
inwards of the black hole. We then consider $\etatil$ the future-directed unit 
normal to $M$. We let $l_+=\etatil-\tnu$ and $l_-=\etatil+\tnu$. By these 
conventions, $l_-$ (resp. $l_+$) is a future-pointing lightlike vector directed inside (resp. outside) of the black hole.
Let $\Theta_\pm$ be the null expansions with respect to $l_\pm$:
$$\Theta_\pm= \bar{g}^{\alpha\beta} \nablah_\alpha(l_\pm)_{\beta},$$
where $\bar{g}$ is the induced metric on $\partial M$, given by 
$$\bar{g}=\gtil-\tnu^\flat \otimes \tnu^\flat,$$

Relatively to $(M,\gtil,K)$, we obtain
\begin{eqnarray*}
 \Theta_\pm &=& \bar{g}^{\alpha\beta} \nablah_\alpha \tn_\beta \mp 
\bar{g}^{\alpha\beta} \nablah_\alpha \tnu_\beta\\
 &=&(\bar{g}^{\alpha\beta}-\tnu^\alpha \tnu^\beta)\nablah_\alpha \tnu_\beta \mp 
H_{\gtil}\\
 &=&\tr_{\gtil} K-K(\tnu , \tnu) \mp H_{\gtil},
\end{eqnarray*}
where $H_{\gtil}$ is the mean curvature of $\partial M$ in $M$ evaluated with 
respect to $\tnu$. $\partial M$ is called a marginally outer trapped surface if $\Theta_+ =0$ and $\Theta_- \leqslant 0$.\\

We want to find initial data $(M, \gtil, K)$ satisfying the constraint equations and such that $\partial M$ is a marginally outer trapped surface. Our constraint equations hence become:

\begin{equation}\label{eqBoundaryConstraints}
\tag{BC}
\begin{cases}
  \scal_{\gtil} -|K|^2_{\gtil} + (\tr_{\gtil} K)^2&=0,\\
  \divg_{\gtil} K-\dd(\tr_{\gtil} K)&=0,\\
  \tr_{\gtil} K-K(\tnu , \tnu) - H &=0,\\
  \tr_{\gtil} K-K(\tnu , \tnu) + H &=\Theta_-.,
\end{cases}
\end{equation}
where the last two equations are imposed on $\partial M$.
As in the case of a compact manifold without boundary, we can resort to the 
very common conformal method to construct 
solutions to \eqref{eqBoundaryConstraints}. Other ways of constructing solutions can be found in \cite{BartIsen}. Recently, variants of the conformal method have been proposed in \cite{DEL16,MaxwellInitialData}.\\

As for the boundary conditions now, if $\nu$ denotes the outward pointing unit normal to $\partial M$ with respect to $g$, we have the conformal change $\nu= \phi^{\frac{N}{2}-1}\tnu$. We also enforce the boundary condition $\sigma (\tnu, \cdot)=0$ for the TT-tensor as it will be explained in Proposition \ref{propYork}.\\

We now compute $K(\tnu,\tnu)$ using the conformal method data:

\begin{align*}
 K(\tnu,\tnu) &= \frac{\tau}{n} \tilg (\tnu ,\tnu) +\varphi^{-2}(\cancel{\sigma (\tnu ,\tnu)} +\Lb_{\tilg} W (\tnu ,\tnu))\\
 &= \frac{\tau}{n}+\varphi^{-N}\Lb_g W(\nu,\nu).
 \end{align*}

We then get
\begin{subequations}
\begin{align}
 0=\Theta_+ &= \frac{n-1}{n}\tau -\varphi^{-N}\Lb_g W (\nu,\nu)-H_{\tilg}, \label{eq1}\\
 \Theta_- &= \frac{n-1}{n}\tau -\varphi^{-N}\Lb_g W (\nu,\nu)+H_{\tilg}. \label{eq2}
\end{align}
\end{subequations}

By substracting both equations \eqref{eq2} and \eqref{eq1}, we obtain $2H_{\tilg}=\Theta_-$. By using the conformal transformation law for  $H_{\tilg}$ (see \cite{G2009} for proof),
\begin{equation}
H_{\tilg}=\varphi^{-\frac{N}{2}+1}\left(H_g+\frac{2(n-1)}{n-2}\frac{\partial_\nu \varphi}{\varphi}\right),
\label{hgtilde}
\end{equation} 
we get a first boundary condition for $\varphi$ related to the Lichnerowicz equation
$$\frac{2(n-1)}{n-2}\partial_\nu \varphi + H_g \varphi = \frac{\Theta_-}{2}\varphi^{\frac{N}{2}}.$$

Concerning the boundary condition for the vector equation, we add \eqref{eq2} and \eqref{eq1} to get

\begin{align*}
 \Theta_- &= 2\frac{n-1}{n}\tau -2\varphi^{-N}\Lb_g W(\nu,\nu) \\
 \Leftrightarrow \Lb_g W(\nu ,\nu) &= \varphi^N \left(\frac{n-1}{n}\tau-\frac{\Theta_-}{2}\right)
\end{align*}

This is a scalar relation, yet we wish to have a vector relation, so to have a reasonable boundary condition for the vector equation, we thus choose to free one of the coordinates by transforming this equality into

$$\mathbb{L}W(\nu, \cdot)=\left(\frac{n-1}{n}\tau -\frac{\Theta}{2}\right) \varphi^N \nu^\flat +\xi,$$

where $\xi$ denotes a 1-form over $\partial M$ we extend on the restriction of $TM$ to $\partial M$ by setting $\xi(\nu)=0$.\\

All the above discussion leads to the \textbf{Boundary Conformal Constraint 
Einstein Equations} \emph{(BCCE for short)}:
  
\begin{subequations}\label{BCCE}
\begin{empheq}[left=(BCCE)\empheqlbrace]{align}
\displaystyle{-\frac{4(n-1)}{n-2}\Delta \phi +\scal~\phi = 
-\frac{n-1}{n}\tau^2 \phi^{N-1} 
+|\sigma+\bL W|^2 \phi^{-N-1}} \label{eqLichnerowicz} \\
\displaystyle{\Delta_{\bL} W = \frac{n-1}{n} \phi^N \dd\tau}\label{eqVector}\\
\displaystyle{\frac{2(n-1)}{n-2}\partial_\nu \phi + H\phi = 
\frac{\Theta_-}{2}\phi^{\frac{N}{2}}}\label{eqCondLich}\\
\displaystyle{(\sigma + \bL W)(\nu, \cdot) = \left(\frac{n-1}{n}\tau 
-\frac{\Theta_-}{2}\right)\phi^{N}\nu^{\flat}
+\xi}\label{eqCondVector}
\end{empheq}
\end{subequations}

We have then a new set of data for the conformal method.

\vspace{0.25cm}
\textbf{Conformal method with boundary conditions:}
\begin{itemize}[label={}]
\item \textbf{Data:} \begin{itemize}[label=\textbullet] \item a Riemannian 
manifold $(M,g)$ with a boundary $\partial M$,
\item a function $\tau : M\rightarrow \bR$,
\item a symmetric 2-tensor $\sigma$ on $M$, traceless, transverse 
($\textrm{div}_g \sigma \equiv 0$), \emph{(TT-tensor)} such that $\sigma(\nu, \cdot) = 0$ on $\partial M$.
\item a function $\Theta_- : \partial M\rightarrow \bR^-$, $\Theta_- \leq 0$,
\item a 1-form $\xi$ belonging to $\Gamma(\partial M, T^*M)$ and such that $\xi(\nu) \equiv 0$.
\end{itemize}
\item \textbf{Unknowns:}
\begin{itemize}[label=\textbullet]
\item a positive function  $\phi:M\rightarrow\bR^+_*$,
\item a vector field $W$.
\end{itemize}
\end{itemize}
  
Before stating the result we will prove, we introduce notations: 

\begin{itemize}[label=-]
\item \emph{The Escobar invariant}. Instead of using the 
Yamabe invariant in further results, we have to generalize it to include the boundary $\partial M$ in its definition. We obtain the Escobar invariant whose properties will be discussed in Section \ref{secEscobar}.

\begin{equation}\label{eqEscobar}
\mathcal{E}(M, g):= \underset{0\not\equiv\phi\in W^{1,2}(M,\bR)}{\textrm{inf}}
\frac{\displaystyle{\int_M \left(\frac{4(n-1)}{n-2}|\dd\phi|^2_g +\scal_g 
        \phi^2\right) \dd \mu^g +2\int_{\partial M}H_g \phi^2 \dd \mu^g}}
{\displaystyle{\left(\int_M\phi^N \dd \mu^g \right)^{\frac{2}{N}}}}.
\end{equation}
  
\item The following notation will be useful in Section \ref{secLich} for truncation purposes. Let $a < b$ be two real numbers.
For any $x \in \bR$, we define
\[
 (x)_{a, b} \definedas \left\lbrace
 \begin{aligned}
  a & \quad \text{if } x < a,\\
  x & \quad \text{if } a \leq x \leq b,\\
  b & \quad \text{if } x > b.
 \end{aligned}
 \right.
\]
When $x$ only needs to be trucated from below or from above, we set
\[
 (x)_{a,} \definedas \max \{a, x\} =
 \left\lbrace
 \begin{aligned}
  a & \quad \text{if } x < a\\
  x & \quad \text{if } a \leq x
 \end{aligned}
 \right.,
 \qquad
 (x)_{,b} \definedas \min \{x, b\} = \left\lbrace
 \begin{aligned}
  x & \quad \text{if } x < b\\
  b & \quad \text{if } b \leq x
 \end{aligned}
 \right..
\]
\end{itemize}
  \subsection{Main result}
  
A first treatment of these equations is made by Dain in \cite{DainTrapped} 
and by Maxwell in \cite{Maxwell}, for vanishing mean curvature (i.e. $\tau \equiv 0$), on asymptotically Euclidean manifolds. The method we are going to use to treat the case of an arbitrary $\tau$ is a variant of the method introduced by Holst, Nagy, Tsogtgerel and Maxwell in \cite{holst2008, holst2009, maxwell2008},
due to Nguyen \cite{cang1} and developed in \cite{GicquaudNguyen, GicquaudSmallTT, GicquaudHabilitation}.
  
The HNTM method requires two things. First, the Yamabe invariant $\mathcal{Y}_g$ has to be positive and second, $\sigma$ has to be non zero but small in a certain sense we will explicit later on. In our context, it is 
necessary to consider not the Yamabe invariant but instead the Escobar invariant, $\mathcal{E}_g$, since our problem is set on a manifold with boundary. Furthermore, in addition to imposing smallness of $\sigma$, we will have to consider smallness assumptions for $\xi$.
  
In a paper published by Gicquaud and Ng\^o in 2014, \cite{GN14}, the same boundary problem we address here is tackled using a method based on the implicit function theorem. Note also that attempts were made to generalize the work by Dahl, Gicquaud, Humbert \cite{dgh} to the system \eqref{BCCE}. However technical difficulties are even worse than the ones we will face in this paper.\\ 

We now recap how the HNTM method works in the case of Maxwell's \cite{maxwell2008} and then, we will discuss T.C.Nguyen's twist to the method, as seen in \cite{cang1}. This is the method exposed in the latter that will be used in this paper, which is an improvement from the classical method.\\   
   
The HNTM method relies on a fixed point theorem, namely a variant of the Schauder theorem (see Theorem \ref{thmSchauder}).

Schematically, if we define
\begin{itemize}
\item $\mathcal{W}_{\tau}$, the application which, for any given positive 
function $\phi$ gives us the associated solution $W$ of the vector equation,
\item $\mathcal{L}_{\sigma, \tau}$ the application which, for any given 
vector field $W$ gives us the associated solution $\phi$ to the Lichnerowicz equation,
\end{itemize}
we can create $\Phi_{\sigma, \tau} := \mathcal{L}_{\sigma, \tau} \circ 
\mathcal{W}_{\tau}$ whose fixed points are solutions of the Lichnerowicz equation that gives us a corresponding $W$, solution of the vector equation; the couple $(\phi, W)$ being a solution to the coupled system.\\
  
In the study of the coupled system, Maxwell demonstrates that if $\varphi^+$ is a global supersolution (\textit{see \cite{holst2009} for definition}), whose existence is assured by \cite[Theorem 5]{holst2009} (see also \cite[Proposition 14]{maxwell2008}), 
there exists a constant $K_0>0$ such that the set $\Omega=\{ \varphi\in L^\infty(M, \Rb) /~K_0\leqslant \varphi\leqslant\varphi^+ \}$ is $\Phi_{\sigma, \tau}$-invariant. $\Omega$ will clearly be convex, closed and bounded in $L^\infty(M, \Rb)$, and it can be shown that $\Phi_{\sigma, \tau}$ is a compact operator by the use of a non-critical Sobolev embedding. This leads to the following result:

\begin{thm}[Theorem 1 of \cite{maxwell2008}]\label{thmMaxwell}
 Let $g\in W^{2,p}$, with $p>n$, be a compact Riemannian manifold. Suppose $g$ has no non-zero conformal Killing vector fields, 
 that $\mathcal{Y}_g>0$, $\sigma\in W^{1,p}$, a non-zero TT-tensor and a function $\tau\in W^{1,p} (M,\Rb)$. 
 
 If $\varphi^{+} \in W^{2,p}$ is a global supersolution for the parameters $(g,\sigma ,\tau)$, then there exists a couple 
 $(\varphi,W)\in W^{2,p}_+ \times W^{2,p}$ solution of the conformal constraint equations \eqref{eqLich}-\eqref{eqVect} with $\varphi \leqslant\varphi^+$.
\end{thm} 
 
Existence of a global supersolution can be proven if $\sigma$ is small enough:

\begin{prop}[Corollary 1 of \cite{maxwell2008}]\label{coroMaxwell}
Let $g\in W^{2,p}$, with $p>n$, be a compact Riemannian manifold. Suppose $g$ has no non-zero conformal Killing vector fields, 
 that $\mathcal{Y}(g)>0$, $\sigma\in W^{1,p}$, a non-zero TT-tensor and a function $\tau\in W^{1,p} (M,\Rb)$. 
 
 If $\sigma\not\equiv 0$ and $\|\sigma \|_{L^\infty}$ is small enough, then there exists a couple 
 $(\varphi,W)\in W^{2,p}_+ \times W^{2,p}$ solution of the conformal constraint equations \eqref{eqLich}-\eqref{eqVect}.
\end{prop}

Combining Theorem \ref{thmMaxwell} and Proposition \ref{coroMaxwell}, we get existence of a solution to the conformal constraint equations provided $\sigma \not\equiv 0$ and $\|\sigma\|_{L^\infty}$ is small enough.
 
The main result of Nguyen in \cite{cang1} was to find lower regularity assumptions under which this result stayed true. He obtained the following result:

\begin{thm}\label{thmnguy}
 Let $(M,g)$ be a compact Riemannian manifold such that $g\in W^{2,{p}}(M,S_2 M)$ for some $\dsp{p>\frac{n}{2}}$ and $\mathcal{Y}(M, g)>0$. Assume also given
 a function $\tau \in W^{1,2p}(M, \Rb)$ and a TT-tensor $\sigma \in W^{1,2p}(M, S_2 M)$, $\sigma\not\equiv 0$.
 Then, there exists $\varepsilon =\varepsilon(M,g,\tau)$ such that, if
 $$\|\sigma\|_{L^2(M, S_2M)} < \varepsilon,$$
 there exists at least one solution $(\varphi, W)\in W^{2,p}(M,\Rb) \times W^{2,p}(M,TM)$ to the conformal constraint equations \eqref{eqLich}-\eqref{eqVect}.
\end{thm} 

This result does not require the use of a global supersolution. It relies instead on a bootstrap argument where we construct nested subsets of $L^\infty(M, \Rb)$ that are for $\Phi_{\sigma, \tau}$ until we reach a bounded one. We will elaborate upon in the following sections.

The main improvement in this theorem compared to Theorem \ref{thmMaxwell} is the fact that $\sigma$ is assumed to be small in $L^2$-norm, not in $L^\infty$-norm meaning that $\sigma$ can be very big in (very) localized parts of $M$.

In this paper, the difference we will encounter here is the boundary terms that will appear while integrating by parts and the subsequent Sobolev spaces to consider. But the main difficulty will be the boundary terms in the elliptic estimates. Indeed, the boundary condition for the vector equation has to be thought as a Neumann condition for $W$. To estimate two derivatives of $W$, we need to estimate one of $\Lb W$ so, in particular, the derivative of $\varphi^N$ as it appears in the right hand side of \eqref{eqCondVector}. This difficulty is what explains why our main result is only valid in small dimension.\\

We can now state the main result we will prove in the following paper:

\begin{thm}\label{thmMain}
Let $(M, g)$ be a compact Riemannian boundary manifold of dimension $3\leqslant n\leqslant 5$ with $g \in W^{2, p}(M, S_2 M)$ 
for a certain $\dsp{p > \frac{n}{2}}$ and $\mathcal{E}(g)>0$.
Suppose given
\begin{itemize}
 \item $\tau \in W^{1, 2p}(M, \Rb)$,
 \item $\Theta_- \in W^{1, 2p}(M, \Rb)$, $\Theta_- \leq 0$ on $\partial M$,
 \item $\sigma \in L^{2p}(M, S_2 M)$ a TT-tensor,
 \item $\xi \in W^{1-\frac{1}{p}, p}(\partial M, TM)$.
\end{itemize}
If $\sigma$ is small enough in
$L^2(M, S_2 M)$-norm and $\xi$ in $L^{n-1}(\partial M, TM)$-norm,
with at least one of them nonzero, then there exists at least one solution $(\varphi, W) \in W^{2, p}(M, \Rb) \times W^{2, p}(M,
TM)$ to the system \ref{BCCE}.
\end{thm}

\subsection*{Outline of the paper}
In section \ref{secLichnerowicz}, we will focus on solving the Lichnerowicz equation \eqref{eqLichnerowicz} with the boundary condition \eqref{eqCondLich} after discussing beforehand some properties of the Escobar invariant. It will be done with a very low regularity context.

Section \ref{secVect} will be devoted to solving the vector equation \eqref{eqVector} with the boundary condition \eqref{eqCondVector}.

Finally, section \ref{secCoupled} will see the use of the Schauder fixed point theorem \ref{thmSchauder} to prove our main result \ref{thmMain}.

\subsection*{Acknowledgements} The author of this paper would like to thank Romain Gicquaud for his great help writing this paper and his great patience and generosity while reviewing preliminary versions of it.
\section{The Lichnerowicz equation}\label{secLichnerowicz}
In this section, we focus on solving for $\phi$ the Lichnerowicz equation
\eqref{eqLichnerowicz} with the boundary condition \eqref{eqCondLich}. To
keep the notation simple, we set $A \definedas |\sigma + \bL W|$.
Construction of solutions to the Lichnerowicz equation on compact manifolds with boundary has been treated in \cite{Dilts, HolstTsogtgerel}. We refer also the reader to the recent article of Sicca and Tsogtgerel \cite{sicca} for the study of the related problem of prescribed scalar and mean curvatures.

The main originality here is to deal with a very low regularity context, i.e. $A\in L^2 (M,\Rb)$ and so $A^2 \in L^1 (M,\Rb)$, which is outside the usual standard elliptic regularity theory. In this section, we will suppose $A\not\equiv 0$.

For convenience, we rewrite here the Lichnerowicz equation \eqref{eqLichnerowicz} with the apparent 
horizon bondary condition \eqref{eqCondLich} with our new notation:

\begin{subequations}\label{eqLichSys}
\begin{empheq}[left=\empheqlbrace]{align}
-\frac{4(n-1)}{n-2} \Delta \phi + \scal_g~\phi + \frac{n-1}{n} \tau^2 
\phi^{N-1} 
&= A^2 \phi^{-N-1}\qquad\text{on }M,\label{eqLichnerowiz2}\\
\frac{2(n-1)}{n-2} \partial_\nu \phi + H_g \phi &= \frac{\Theta_-}{2} 
\phi^{\frac{N}{2}}\qquad\text{on } \partial M.\label{eqCondLich2}
\end{empheq}
\end{subequations}

Let $\dsp{p > \frac{n}{2}}$ be given. The regularity assumptions we will make in this section are the following:
\begin{itemize}
 \item $g \in W^{2, p}(M, S_2M)$,
 \item $\tau \in L^\infty(M, \bR)$,
 \item $H_g \in W^{1-\frac{1}{p}, p}(\partial M, \bR)$,
 \item $\Theta_- \in W^{1-\frac{1}{p}, p}(\partial M, \bR) 
\cap L^\infty(\partial M, \bR)$.
\end{itemize}

Note that, from the Sobolev embedding theorem (see \cite{LeoniSobolev}), we have
$H_g \in L^q(\partial M, \bR)$, with $\dsp{\frac{1}{q} = \frac{1}{p} - 
\left(1-\frac{1}{p}\right) \frac{1}{n-1}}$,
so, in particular, $q > n-1$.

\subsection{The Escobar invariant}\label{secEscobar}
As the Sobolev and trace inequalities will play an important role in what
follows, we recall them here (see e.g. \cite{LeoniSobolev}):
\begin{itemize}
\item \textsc{Sobolev inequality:} there exists a constant
$s_0 = s_0(M, g)$ so that, for all $u \in W^{1, 2}(M, \bR)$, we have
\[
 \|u\|_{L^N(M, \bR)}^2 \leq s_0 \|u\|_{W^{1, 2}(M, \bR)}^2,
\]
\item \textsc{Trace inequality:} there exists a continuous operator
$W^{1, 2}(M, \bR) \to L^{\fraccheloue}(\partial M, \bR)$ that coincide with the 
classical
restriction operator $u \mapsto u\vert_{\partial M}$ for $C^1$-functions.
As a consequence, there exists a constant $t_0$ such that, for all
$u \in W^{1, 2}(M, \bR)$, we have
\[
 \|u\|_{L^{\fraccheloue}(\partial M, \bR)}^2 \leq t_0 \|u\|_{W^{1, 2}(M, \bR)}^2.
\]
\end{itemize}

We remind the reader that the Escobar invariant was given in \eqref{eqEscobar}.
Due to our regularity requirements, this invariant is well-defined.
As a shorthand, we will denote by $Q_g(u)$ the quadratic form that appears as 
the numerator:
\begin{equation}\label{eqDefQuadrForm}
 Q_g(u) \definedas \frac{1}{2} \int_M \left[\frac{4(n-1)}{n-2} |\dd u|^2 + 
\scal_g u^2\right] \dd\mu^g + \int_{\partial M} H_g u^2 \dd\mu^g.
\end{equation}

We also let
\begin{equation}\label{eqEigenvalue}
 \lambda_0(M, g) \definedas \inf_{u \in W^{1, 2}(M, \bR), u \not \equiv 0} 
\frac{\dsp{\frac{1}{2} \int_M \left[\frac{4(n-1)}{n-2} |\dd u|^2 + \scal_g~u^2\right] 
\dd\mu^g + \int_{\partial M} H_g u^2 \dd\mu^g}}{\dsp{\int_M u^2 \dd\mu^g}}
\end{equation}
denote the first eigenvalue $\lambda$ for the conformal Laplacian with Robin 
boundary condition:
\[
\left\lbrace
\begin{aligned}
 -\frac{4(n-1)}{n-2} \Delta u + \scal_g~u &= \lambda u \qquad\text{on }M,\\
 \frac{2(n-1)}{n-2} \partial_\nu u + H_g u &= 0 \qquad\text{on }\partial M.
\end{aligned}
\right.
\]

As $(M, g)$ has finite volume, we have, for all $u \in L^N(M, \bR)$,
\[
 \int_M u^2 \dd\mu^g \leq \vol(M, g)^{2/n} \left(\int_M u^N \dd\mu^g\right)^{2/N}.
\]
This implies that, assuming that $\cE(M, g) > 0$,
\begin{align}\label{eqEigenvalueIneq}
 \lambda_0(M, g) &\geq \inf_{u \in W^{1, 2}(M, \bR), u \not \equiv 0} 
\frac{Q_g(u)}{\dsp{\vol(M, g)^{2/n} \left(\int_M u^N \dd\mu^g\right)^{2/N}}}\nonumber\\
 &= \frac{1}{\vol(M, g)^{2/n}} \cE(M, g) > 0.
\end{align}

\begin{lemma}\label{lmEscobar}
We have the equivalence
\[
 \cE(M, g) > 0 \Leftrightarrow Q_g\text{ is coercive}.
\]
\end{lemma}
It should be noted that coercivity for quadratic forms (on $W^{1, 2}(M,\Rb)$) is 
equivalent to the following : there exists $\varepsilon > 0$ so that, for all $u 
\in W^{1, 2}(M, \bR)$,
we have
\[
 Q_g(u) \geq \varepsilon \|u\|^2_{W^{1, 2}(M, \bR)}.
\]

\begin{proof}
The proof is based on applications of the H\"older inequality. We write
\begin{align*}
 \int_M \scal_g~u^2 \dd\mu^g
  &\geq - \int_M |\scal_g|~|u|^\gamma |u|^{2-\gamma} \dd\mu^g\\
  &\geq -\|\scal_g\|_{L^p(M, \bR)} \| |u|^\gamma\|_{L^a(M, \bR)} \| 
|u|^{2-\gamma}\|_{L^b(M, \bR)}\\
  &= -\|\scal_g\|_{L^p(M, \bR)} \|u\|_{L^{\gamma a}(M, \bR)}^\gamma 
\|u\|_{L^{(2-\gamma)b}(M, \bR)}^{2-\gamma}.
\end{align*}
with $\gamma, a, b$ constant so that $a, b \in [1, \infty]$ satisfy
\[
 \frac{1}{p} + \frac{1}{a} + \frac{1}{b} = 1.
\]
Choosing $\dsp{\gamma = 2 - \frac{n}{p} \in (0, 2)}$, $\dsp{a = \frac{2}{\gamma}}$ and $\dsp{b = \frac{N}{2-\gamma}}$,
we get
\begin{equation}\label{eqEstimate1}
\begin{aligned}
 \int_M \scal_g~u^2 \dd\mu^g
 &= -\|\scal_g\|_{L^p(M, \bR)} \|u\|_{L^2(M, \bR)}^\gamma \|u\|_{L^N(M, 
\bR)}^{2-\gamma}\\
 &\geq - \varepsilon_1 \|u\|_{L^N(M, \bR)}^2 - c_1 \|u\|_{L^2(M, \bR)}^2,
\end{aligned}
\end{equation}
where we used Young's inequality with an arbitrary parameter $\varepsilon_1 > 0$. 
The value of $c_1$ is explicit (and depends on $\varepsilon_1$) but will not be
useful in what follows.

Similarly, we have
\[
 \int_{\partial M} H_g u^2 \dd\mu^g \geq - \|H_g\|_{L^q(\partial M, \bR)} 
\|u\|_{L^r(\partial M, \bR)}^2,
\]
with $r$ such that $\dsp{\frac{1}{q} + \frac{2}{r} = 1}$ so
$\dsp{r \in \left(2, \frac{N}{2}+1\right)}$. We claim that for any $\varepsilon_2 > 0$,
there exists a constant $c_2$ such that
\begin{equation}\label{eqEstimate2}
 \|u\|_{L^r(\partial M, \bR)}^2 \leq \varepsilon_2 \|u\|_{W^{1, 2}(M, \bR)}^2 + 
c_2 
\|u\|_{L^2(M, \bR)}^2.
\end{equation}
The proof goes by contradiction. Assume that there exists an $\varepsilon_2 > 0$
so that \eqref{eqEstimate2} fails for all $c_2 > 0$. There exists a
sequence $(u_k)_k$ of non-zero functions satisfying, for all $k\in\mathbb{N}$,
\begin{equation}\label{eqEstimate3}
 \|u_k\|_{L^r(\partial M, \bR)}^2 \geq \varepsilon_2 \|u_k\|_{W^{1, 2}(\partial M, 
\bR)}^2 + k \|u_k\|_{L^2(M, \bR)}^2.
\end{equation}
Multiplying each $u_k$ by some constant, we can assume that
$\|u_k\|_{L^r(\partial M, \bR)} = 1$ for all $k$. We then have
\[
 \|u_k\|_{W^{1, 2}(M, \bR)}^2 \leq \varepsilon_2^{-1},
\]
so $(u_k )_k$ is bounded in $W^{1,2}(M,\Rb)$.
Since the embedding $W^{1, 2}(M,\bR) \into L^2(M, \bR)$ is compact, as
is the trace operator $W^{1, 2}(M, \bR) \to L^r(\partial M, \bR)$, upon extracting a subsequence, we can assume that
$(u_k)_k$ converges weakly in $W^{1, 2}(M, \bR)$ to some function $u_\infty$,
strongly in $L^2(M, \bR)$ to $v_\infty$ and the traces $u_k\vert_{\partial M}$
converge strongly in $L^r(\partial M, \bR)$ to $w_\infty$.

For any $f \in L^2(M, \bR)$, we have
\[
 \int_M f u_k \dd\mu^g \to \int_M f v_\infty \dd\mu^g,
\]
but, as $\dsp{v \mapsto \int_M f v \dd\mu^g}$ is a linear form on $W^{1, 2}(M, \bR)$,
we also have 
\[
 \int_M f u_k \dd\mu^g \to \int_M f u_\infty \dd\mu^g.
\]
As this is true for any $f$, this shows that $u_\infty \equiv v_\infty$.
From \eqref{eqEstimate3}, it follows that $\dsp{\|u_k\|_{L^2(M, \bR)}^2 \leq \frac{1}{k}}$
so $u_\infty \equiv 0$. As $v \mapsto \int_{\partial M} w_\infty v \dd\mu^g$ is
a continuous linear form on $W^{1, 2}(M, \bR)$, we have
\[
 0 = \int_{\partial M} w_\infty u_\infty \dd\mu^g = \lim_{k \to \infty} 
\int_{\partial M} w_\infty u_k \dd\mu^g = \int_{\partial M} w_\infty^2 \dd\mu^g.
\]
This shows that $w_\infty \equiv 0$. But, as $\|u_k\|_{L^r(\partial M, \bR)} = 
1$,
we have $\|w_\infty\|_{L^r(\partial M, \bR)} = 1$. This gives the desired
contradiction. Then, there exists a constant $c_2$ such that \eqref{eqEstimate1} is satisfied, for any $u\in W^{1,2}(M,\bR)$.\\

Inserting \eqref{eqEstimate2} and \eqref{eqEstimate3} in the definition of 
$Q_g$, we get that
\begin{align*}
 Q_g(u)
  &\geq \frac{2(n-1)}{n-2} \int_M |\dd u|^2 \dd\mu^g - \left(\frac{\varepsilon_1}{2} + \varepsilon_2 \|H_g\|_{L^q(\partial M, \bR)} \right) \|u\|_{W^{1, 2}(M, \bR)}^2\\
  &\qquad\qquad - \left(\frac{c_1}{2} + \|H_g\|_{L^q(\partial M, \bR)} c_2 \right) \|u\|_{L^2(M, 
\bR)}^2.
\end{align*}
As $\varepsilon_1$ and $\varepsilon_2$ can be chosen as close to zero as we want,
we can assume that
\[
 \frac{2(n-1)}{n-2} - \left(\frac{\varepsilon_1}{2} + \varepsilon_2 \|H_g\|_{L^q(\partial M, \bR)} \right)
 \geq 1.
\]
Thus, we have proven that there exists a constant $C$ such that
\begin{equation}\label{eqEstimate4}
 Q_g(u) \geq \|u\|_{W^{1, 2}(M, \bR)}^2 - C \|u\|_{L^2(M, \bR)}^2.
\end{equation}
Finally, we introduce a parameter $\alpha \in (0, 1)$ and write
\begin{align*}
 Q_g(u) &= \alpha Q_g(u) + (1-\alpha) Q_g(u)\\
 & \geq \alpha \left(\|u\|_{W^{1, 2}(M, \bR)}^2 - C \|u\|_{L^2(M, 
\bR)}^2\right) 
+ (1-\alpha) \lambda_0(M, g) \|u\|_{L^2(M, \bR)}^2.
\end{align*}
Choosing $\dsp{\alpha = \frac{\lambda_0(M, g)}{C + \lambda_0(M, g)} > 0}$, so
that $\alpha C = (1-\alpha) \lambda_0(M, g)$, we finally get
that, for all $u \in W^{1, 2}(M, \bR)$,
\[
 Q_g(u) \geq \alpha \|u\|_{W^{1, 2}(M, \bR)}^2.
\]
This is the desired coercivity estimate.
\end{proof}

\begin{proposition}\label{propEigenFunction}
Assume that $\cE(M, g) > 0$.
The space of functions $u \in W^{1, 2}(M, \bR)$ such that $Q_g(u) = 
\lambda_0(M, 
g) \|u\|^2_{L^2(M,  \bR)}$
is one dimensional. It is generated by a (strictly) positive function
$u_0 \in W^{2, p}(M, \bR)$ such that
\begin{equation}\label{eqEigenfunction}
\left\lbrace
\begin{aligned}
 -\frac{4(n-1)}{n-2} \Delta u_0 + \scal_g~u_0 &= \lambda_0(M, g) u_0 \qquad\text{on }M,\\
 \frac{2(n-1)}{n-2} \partial_\nu u_0 + H_g u_0 &= 0 \qquad\text{on }\partial M.
\end{aligned}
\right.
\end{equation}
\end{proposition}
Remark that this proposition remains true without assuming $\cE(M, g) > 0$.
As this will not be needed in the sequel, we restrict to this slightly
simpler case.

\begin{proof}
By the definition of $\lambda_0(M, g)$, there exists a sequence of
functions $u_k \not\equiv 0$ such that
$$
 \frac{Q_g(u_k)}{\|u_k\|^2_{L^2(M, \bR)}} \to \lambda_0(M, g).
$$
Without loss of generality, we can assume that $\|u_k\|^2_{L^2(M, \bR)} = 1$.
And, since $Q_g(|u_k|) = Q_g(u_k)$, we can also assume that $u_k \geq 0$.
By the coercivity of $Q_g$ (see Lemma \ref{lmEscobar}), we get that the
sequence $(u_k)_k$ is bounded in $W^{1, 2}(M, \bR)$. As a consequence,
we can assume, without loss of generality that $(u_k)_k$ converges
to a certain $u_\infty \in W^{1, 2}(M, \bR)$ strongly in $L^2(M, \bR)$
and weakly in $W^{1, 2}(M, \bR)$. This implies that
$\|u_\infty\|_{L^2(M, \bR)} = 1$ and, as $Q_g$ is convex,
$Q_g(u_\infty) \leq \liminf_{k \to \infty} Q_g(u_k) = \lambda_0(M, g)$. By definition of $\lambda_0(M, g)$, we have $Q_g(u_\infty) = \lambda_0(M, g)$.

For each $v \in W^{1, 2}(M, \bR)$, the mapping
$$
 t \mapsto Q_g(u_\infty + tv) - \lambda_0 (M,g) \|u_\infty + tv\|^2_{L^2(M, \bR)}
$$
is differentiable and reaches its minimum value at $t=0$. This shows that $u_\infty$
solves
\begin{equation}\label{eqEigenFunctionWeak}
0 = \int_M \left[\frac{4(n-1)}{n-2} \<du_\infty, dv\> + \scal_g~u_\infty v - 
\lambda_0 u_\infty v\right] \dd\mu^g + 2 \int_M H_g u_\infty v \dd\mu^g,
\end{equation}
which is the weak formulation of Equation \eqref{eqEigenfunction}.
By elliptic regularity, we have that $u_\infty$ belongs to $W^{2, p}(M,\bR)$ and 
solves \eqref{eqEigenfunction} in the strong sense (see \cite[Lemma 4.5]{sicca} for the detail of the proof, see also Proposition \ref{propLichStrong} for a similar argument). From \cite[Lemma 
B.7]{HolstTsogtgerel}, we conclude that $u_\infty$ is continuous and strictly 
positive.

Assume now that there exists a function $v \in W^{1, 2}(M, \bR)$ linearly
independent of $u$ such that $Q_g(v) = \lambda_0(M, g) \|v\|_{L^2(M, \bR)}$.
By the same reasoning, we have that $v \in W^{2, p}(M, \bR)$ satisfies
\eqref{eqEigenfunction}. As $u_\infty$ and $v$ are both continuous, we
can assume, by rescaling $v$ that $-u_\infty \leq v \leq u_\infty$.
Consider the family of functions $v_t = v+tu_\infty$. If $t \leq -1$,
we have $v_t \leq 0$ and if $t \geq 1$, $v_t \geq 0$. We set
$t_0 = \min \{t, v_t \geq 0\}$. Then $v_{t_0} \geq 0$ and there exists
a point $x \in M$ such that $v_{t_0}(x) = 0$. As $v_{t_0} \not\equiv 0$
(since $u_\infty$ and $v$ are linearly independent), this gives a
contradiction with the strong maximum principle
\cite[Lemma B.7]{HolstTsogtgerel}.
\end{proof}

The interest of the Escobar invariant is its conformal invariance, i.e. if $g$ and $g_0$ are two conformally equivalent metrics, we get $\mathcal{E}(M,g) =\mathcal{E}(M,g_0)$. Seek \cite[Lemma 2.3]{sicca} for more details. Later on, we will only need the special case $g_0 := u_0^{N-2} g$ (for example in the proof of Proposition \ref{propImprovedRegL2}), which we prove below.  
By using the relation \eqref{hgtilde} and the following, proved in particular in \cite[Chapter 1, Section 7.2]{chow2006}:

\begin{equation}    
    \scal_{\tilg}=\left(-\frac{4(n-2)}{n-1}\Delta_g \varphi +\scal_g \varphi\right)\varphi^{1-N}
    \label{transconfscal}
    \end{equation}

Equations \eqref{eqEigenfunction} give us
\begin{equation}\label{eqConfTransf}
 \scal_{g_0} \equiv \lambda_0 (M,g) u_0^{2-N}, \qquad H_{g_0} \equiv 0.
\end{equation}

Finally, for any $\ubar \in W^{1, 2}(M, \bR)$, we compute
\begin{align*}
Q_g(u_0 \ubar)
 &= \frac{1}{2} \int_M \left[\frac{4(n-1)}{n-2} |\dd(u_0 \ubar)|_g^2 + 
\scal_g~(u_0 \ubar)^2\right] \dd\mu^g + \int_{\partial M} H_g (u_0 \ubar)^2 
\dd\mu^g\\
 &= \frac{1}{2} \int_M \left[\frac{4(n-1)}{n-2} \left(u_0^2 |\dd\ubar|_g^2 + 
\<u_0 
\dd u_0, \dd(\ubar)^2\> + \ubar^2 |\dd u_0|^2_g\right) + \scal_g~(u_0 \ubar)^2\right] 
\dd\mu^g\\
 &\qquad + \int_{\partial M} H_g (u_0 \ubar)^2 \dd\mu^g\\
  &= \frac{1}{2} \int_M \left[\frac{4(n-1)}{n-2} \left(u_0^2 |\dd\ubar|_g^2 - 
\ubar^2 u_0 \Delta_g u_0\right) + \scal_g~(u_0 \ubar)^2\right] \dd\mu^g\\
 &\qquad + \int_{\partial M} \left[\frac{2(n-1)}{n-2} \ubar^2 u_0 \partial_\nu 
u_0 + H_g (u_0 \ubar)^2\right] \dd\mu^g\\
 &= \frac{1}{2} \int_M \left[\frac{4(n-1)}{n-2} u_0^2 |\dd\ubar|_g^2 + 
\lambda_0(M, g) u_0^2 \ubar^2 \right] \dd\mu^g\\
 &= \frac{1}{2} \int_M \left[\frac{4(n-1)}{n-2} |\dd\ubar|_{g_0}^2 + \lambda_0(M, 
g) u_0^{2-N} \ubar^2 \right] \dd\mu^{g_0}.\\
\end{align*}
Hence, we have proven that
\begin{equation}\label{eqConfTransQ}
 Q_g(u_0 \ubar) = Q_{g_0}(\ubar).
\end{equation}

\subsection{Solving the Lichenrowicz equation}\label{secLich}
Our first aim in this section is to prove the following result:

\begin{proposition}\label{propLich}
Assume $A \in L^2(M, \bR)$, $A \not\equiv 0$ is given and $g$ satisfies $\mathcal{E}(M,g)>0$. Then there exists a
unique positive function $\phi \in W^{1, 2}(M, \bR)$ satisfying the following
weak formulation of \eqref{eqLichSys}: for all $\psi \in W^{1, 2}(M, \bR) \cap
L^\infty(M, \bR)$,
\begin{equation}\label{eqWeakForm0}
\begin{aligned}
0 &= \int_M \left[\frac{4(n-1)}{n-2} \<\dd\phi, \dd\psi\> + \scal_g~\phi\psi 
+ 
\frac{n-1}{n} \tau^2 \phi^{N-1} \psi - \frac{A^2 \psi}{\phi^{N+1}}\right] 
\dd\mu^g\\
  &\qquad + 2 \int_{\partial M} \left[H_g \phi - \frac{\Theta_-}{2} 
\phi^{\frac{N}{2}} 
\right]\psi \dd\mu^g.
\end{aligned}
\end{equation}
Furthermore, $\phi$ is uniformly bounded from below and the mapping $A \mapsto \phi$ 
is continuous as a map from $L^2(M, \bR)$ to $W^{1, 2}(M, \bR)$.
\end{proposition}

Several other existence results for the Lichnerowicz equation on a compact manifold with boundary are stated in \cite{Dilts, HolstTsogtgerel}. Part of 
the ideas are borrowed from \cite{GicquaudNguyen,GicquaudSmallTT}. However, 
new arguments are needed to handle regularity at the boundary and develop a 
low regularity framework.\\

To prove the proposition, we introduce a family of functionals over $W^{1, 
2}(M, \bR)$ depending on a parameter $\mu \in [0, \infty)$:
\begin{equation}\label{eqDefImu}
\begin{aligned}
 I_\mu(\phi)
  &\definedas \int_M \left[\frac{1}{2} \left(\frac{4(n-1)}{n-2} |\dd\phi|^2 + 
\scal_g~\phi^2\right) + \frac{1}{N}\left(\frac{n-1}{n} \tau^2 \phi^N + 
\frac{A^2}{(\phi+\mu)^N} \right)\right] \dd\mu^g\\
  &\qquad + \int_{\partial M} \left[H_g \phi^2 - \frac{2}{N+2} \Theta_- 
\phi^{\fraccheloue} \right] \dd\mu^g\\
  & = Q_g(\phi) + \frac{1}{N}\int_M \left(\frac{n-1}{n} \tau^2 \phi^N + 
\frac{A^2}{(\phi+\mu)^N} \right) \dd\mu^g - \frac{2}{N+2} \int_{\partial M} 
\Theta_- \phi^{\fraccheloue} \dd\mu^g.
\end{aligned}
\end{equation}
It should be noted that $I_\mu$ is not defined over all $W^{1, 2}(M, \bR)$
but at least over the subset
\[
 C_{\mu/2} \definedas \{\phi \in W^{1, 2}(M, \bR), \phi \geq 
-\mu/2\text{ 
a.e.}\}
\]
(the full domain of definition is larger but this will not be
relevant here). $C_{\mu/2}$ is convex and closed subset as it can
be defined as
\[
 C_{\mu/2} = \bigcap_{\substack{f \in L^2(M, \bR),\\ f \geq 0~a.e.}} 
\left\{\phi \in W^{1, 2}(M, \bR), \int_M f \phi \dd\mu^g \geq 
-\frac{\mu}{2} 
\int_M f \dd\mu^g\right\}.
\]
The functional $I_0$ will play an important role as
its critical points will be solutions to equations \eqref{eqLichSys}.

We remark that $I_\mu(\phi) \geq Q_g(\phi)$ for all $\phi \in W^{1, 
2}(M, \bR)$.
As a consequence, from Lemma \ref{lmEscobar}, we conclude that $I_\mu$ is
coercive.

\begin{lemma}\label{lmSemicont}
For all $\mu > 0$ the functional $I_\mu$ is sequentially weakly lower
semicontinuous on $C_{\mu/2}$, i.e. for any sequence $(\phi_k)_k$
weakly converging to some $\phi_\infty$, we have
\[
 I_\mu(\phi_\infty) \leq \liminf_{k \to \infty} I_\mu(\phi_k).
\]
\end{lemma}

\begin{proof}
$Q_g$ is strongly continuous and coercive, hence convex. In particular, it
is weakly lower semicontinuous. The same is true for the maps
\[
 \phi \mapsto \frac{1}{N} \frac{n-1}{n} \int_M  \tau^2 \phi^N 
\dd\mu^g\quad\text{and}\quad \phi \mapsto - \frac{2}{N+2} \int_{\partial M} 
\Theta_- \phi^{\fraccheloue} \dd\mu^g.
\]
All that we need to check is that
\[
 J_\mu: \phi \mapsto \frac{1}{N}\int_M \frac{A^2}{(\phi+\mu)^N} \dd\mu^g
\]
is sequentially weakly semicontinuous. Assume by contradiction that there exists
a sequence $(\phi_k)_k$ in $C_{\mu/2}$ weakly converging to
$\phi_\infty$ such that
$J_\mu(\phi_k) \not\to_{k \to \infty} J_\mu(\phi_\infty)$.

Upon passing to a subsequence, we can assume that for some $\varepsilon > 0$,
we have $|J_\mu(\phi_k) - J_\mu(\phi_\infty)| \geq \varepsilon$.
We can also assume that $\phi_k \to \phi_\infty$ in $L^2(M, \bR)$, for all $k\geq 0$ (see the
proof of Lemma \ref{lmEscobar}). In particular, we have
$\phi_k \to \phi_\infty$ a.e.. Now remark that, for all $k \geq 0$, we have
\[
 \frac{A^2}{(\phi_k + \mu)^N} \leq \left(\frac{2}{\mu}\right)^N A^2.
\]
As a consequence, it follows from the dominated convergence theorem that
\[
 J_\mu(\phi_k) = \frac{1}{N}\int_M \frac{A^2}{(\phi_k+\mu)^N} \dd\mu^g 
\to_{k \to 
\infty} \frac{1}{N}\int_M \frac{A^2}{(\phi_\infty+\mu)^N} \dd\mu^g = 
J_\mu(\phi_\infty).
\]
This contradicts the assumption that
$|J_\mu(\phi_k) - J_\mu(\phi_\infty)| \geq \varepsilon$ and give the desired
contradiction: $J_\mu$ is sequentially weakly continuous.
\end{proof}

Since $I_\mu$ is sequentially weakly lower semicontinuous and coercive,
it admits a minimum $\phi_\mu$ over $C_{\mu/2}$. Since
$I_\mu(|\phi|) \leq I_\mu(\phi)$ for all $\phi \in W^{1, 2}(M, \bR)$, 
we can
assume that $\phi_\mu \geq 0$ a.e.

As we only assume that $A^2 \in L^1(M, \bR)$, $I_\mu$ is not differentiable
in the sense of Fréchet. Still, if $\phi \in C_{\mu/2}$
and if $\psi \in W^{1, 2}(M, \bR) \cap L^\infty(M, \bR)$, we have,
for all $\lambda$ small enough $\dsp{\left( |\lambda| \leq \frac{\mu}{4} 
\|\psi\|_{L^\infty (M,\Rb)}\right)}$,
\begin{align*}
&\frac{1}{N}\int_M \frac{A^2}{(\phi + \lambda \psi + \mu)^N} \dd\mu^g - 
\frac{1}{N}\int_M \frac{A^2}{(\phi+\mu)^N} \dd\mu^g\\
&\qquad = \frac{1}{N}\int_M A^2 \left[\frac{1}{(\phi + \lambda \psi + \mu)^N} 
- 
\frac{1}{(\phi + \mu)^N}\right] \dd\mu^g\\
&\qquad = -\lambda \int_M \int_0^1 \frac{\psi A^2}{(\phi + \theta\lambda 
\psi+\mu)^{N+1}} \dd\theta \dd\mu^g.
\end{align*}
Note that due to our upper bound for $\lambda$, the last integrand has its absolute value bounded from above by
$$\dsp{\left(\frac{\mu}{4}\right)^{-N+1} A^2 |\psi| \in L^1(M, \bR)}$$ 
and converges a.e. to
$\psi A^2 \phi^{-N-1}$ as $\lambda$ tends to zero.
As a consequence, from the dominated convergence theorem, we conclude
that
\begin{align*}
&\frac{1}{N}\int_M \frac{A^2}{(\phi + \lambda \psi + \mu)^N} \dd\mu^g - 
\frac{1}{N}\int_M \frac{A^2}{(\phi+\mu)^N} \dd\mu^g\\
&\qquad = -\lambda \int_M \frac{\psi A^2}{(\phi + \mu)^{N+1}} \dd\mu^g + 
o(\lambda).
\end{align*}
This shows that the term $J_\mu(\phi)$ admits directional derivatives
in the direction of $\psi \in W^{1, 2}(M, \bR) \cap L^\infty(M, \bR)$.
By standard arguments, we get that
\begin{equation}\label{eqWeakForm}
\begin{aligned}
0 &= \frac{\dd}{\dd \lambda} I_\mu (\phi_\mu +\lambda \psi) \vert_{\lambda =0}\\ 
&= \int_M \left[\frac{4(n-1)}{n-2} \<\dd\phi_\mu, \dd\psi\> + 
\scal_g~\phi_\mu\psi 
+ \frac{n-1}{n} \tau^2 \phi_\mu^{N-1} \psi - \frac{A^2 
\psi}{(\phi_\mu+\mu)^{N+1}}\right] \dd\mu^g\\
  &\qquad + 2 \int_{\partial M} \left[H_g \phi_\mu - \frac{\Theta_-}{2} 
\phi_\mu^{N/2} \right]\psi \dd\mu^g.
\end{aligned}
\end{equation}
If $\phi_\mu \in W^{2, 2}(M, \bR)$ (which will be proven in Proposition 
\ref{propLichStrong} with $\mu = 0$), we can perform an integration by parts and get that, 
for all $\psi \in W^{1, 2}(M, \bR) \cap L^\infty(M, \bR)$,
\[
\begin{aligned}
0 &= \int_M \left[-\frac{4(n-1)}{n-2} \Delta \phi_\mu + \scal_g~\phi_\mu 
+ 
\frac{n-1}{n} \tau^2 \phi_\mu^{N-1} - 
\frac{A^2}{(\phi_\mu+\mu)^{N+1}}\right] 
\psi \dd\mu^g\\
  &\qquad + 2 \int_{\partial M} \left[\frac{2(n-1)}{n-2} \partial_\nu \phi + 
H_g 
\phi_\mu - \frac{\Theta_-}{2} \phi_\mu^{N/2} \right]\psi \dd\mu^g.
\end{aligned}
\]
This shows that \eqref{eqWeakForm} is a weak formulation of the following
modification of the problem \eqref{eqLichSys}:
\begin{subequations}\label{eqLichSysMu}
\begin{empheq}[left=\empheqlbrace]{align}
-\frac{4(n-1)}{n-2} \Delta \phi_\mu + \scal_g~\phi_\mu
 &= - \frac{n-1}{n} \tau^2 \phi_\mu^{N-1} +
\frac{A^2}{(\phi_\mu+\mu)^{N+1}},\label{eqLichMu}\\
\frac{2(n-1)}{n-2} \partial_\nu \phi_\mu + H_g \phi_\mu &=  
\frac{\Theta_-}{2} 
\phi_\mu^{\frac{N}{2}}.\label{eqCondLichMu}
\end{empheq}
\end{subequations}

We now let $\mu$ tend to zero. Before this, we have to prove that
$\phi_\mu$ is uniformly bounded from below. This will be done by
constructing a suitable subsolution to the system \eqref{eqLichSys}.

\begin{lemma}\label{lmSubSol}
Let $\Lambda > 0$, there exists a unique solution
$u \in W^{2, p}(M, \bR)$ to the following Robin problem:
\begin{subequations}\label{eqLinSys}
\begin{empheq}[left=\empheqlbrace]{align}
-\frac{4(n-1)}{n-2} \Delta u + \scal_g u + \frac{n-1}{n} \tau^2 u &= 
(A^2)_{,\Lambda} \qquad \text{on }M,\label{eqLinLich}\\
\frac{2(n-1)}{n-2} \partial_\nu u + H_g u &= - u \qquad \text{on }\partial M.\label{eqLinCondLich}
\end{empheq}
\end{subequations}
The function $u$ is strictly positive.
\end{lemma}

\begin{proof}
Note that $(A^2)_{,\Lambda} \in L^\infty(M, \bR) \subset L^p(M, \bR)$.
The solution $u$ to \eqref{eqLinSys} exists and belongs to $W^{2, p}(M, \bR)$.
Indeed, from general theory (see e.g. \cite[Lemma B.6]{HolstTsogtgerel},
the associated operator
\[
\begin{array}{ccc}
 W^{2, p}(M, \bR) &\to & L^p(M, \bR) \times W^{1-\frac{1}{p}, p}(\partial M, \bR)\\
 u & \mapsto & \left(\dsp{-\frac{4(n-1)}{n-2} \Delta u + \scal_g~u + \frac{n-1}{n} 
\tau^2 u,~\frac{2(n-1)}{n-2} \partial_\nu u + (H_g+1) u}\right)
\end{array}
\]
is Fredholm with index zero. As a consequence, all we need to check is
that the only solution to the associated homogeneous problem is $v \equiv 0$.
However, if $v$ solves the homogeneous problem
\begin{subequations}\label{eqLinSysH}
\begin{empheq}[left=\empheqlbrace]{align}
-\frac{4(n-1)}{n-2} \Delta v + \scal_g v + \frac{n-1}{n} \tau^2 v &= 0 \qquad \text{on }M,
\label{eqLinLichH}\\
\frac{2(n-1)}{n-2} \partial_\nu v + H_g v &= - v \qquad \text{on }\partial M,\label{eqLinCondLichH}
\end{empheq}
\end{subequations}
we have, multiplying \eqref{eqLinLichH} by $v$, integrating over $M$
and using the boundary condition \eqref{eqLinCondLichH},
\begin{align*}
 0
 &= \int_M \left[-\frac{4(n-1)}{n-2} v \Delta v + \scal_g~v^2 + \frac{n-1}{n} 
\tau^2 v^2\right] \dd\mu^g\\
 &= \int_M \left[\frac{4(n-1)}{n-2} |\dd v|^2 + \scal_g~v^2 + \frac{n-1}{n} \tau^2 
v^2\right] \dd\mu^g - \frac{4(n-1)}{n-2} \int_{\partial M} v \partial_\nu v 
\dd\mu^g\\
 &= \int_M \left[\frac{4(n-1)}{n-2} |\dd v|^2 + \scal_g~v^2 + \frac{n-1}{n} \tau^2 
v^2\right] \dd\mu^g + 2 \int_{\partial M} (H_g + 1) v^2 \dd\mu^g\\
 &= 2 Q_g(v) + \frac{n-1}{n} \int_M \tau^2 v^2 + \int_{\partial M} v^2 \dd\mu^g.
\end{align*}
Since, by assumption, all terms on the last line are non-negative, we have
$Q_g(v) = 0$ which implies $v \equiv 0$.

To show that $u$ is strictly positive, we reintroduce the metric $g_0 = 
u_0^{N-2} g$, where $u_0$ was introduced in Proposition 
\ref{propEigenFunction} and set $\dsp{\ubar = \frac{u}{u_0}}$. From the conformal 
transformation law of the conformal Laplacian, we get
that $\ubar$ solves
\[
 \left\lbrace
 \begin{aligned}
  -\frac{4(n-1)}{n-2} \Delta_{g_0} \ubar + \lambda_0 u_0^{2-N}~\ubar + 
  \frac{n-1}{n} \tau^2 u_0^N \ubar &= u_0^{N-1} (A^2)_{,\Lambda} \qquad \text{on }M,\\
  \frac{2(n-1)}{n-2} \partial_{\nubar} \ubar +  u_0^{\frac{N}{2}+1} \ubar&= 0 \qquad \text{on }\partial M.
 \end{aligned}
 \right.
\]
As $(A^2)_{,\Lambda} \geq 0$ is non trivial, we conclude from \cite[Lemma 
B.7]{HolstTsogtgerel} that $\ubar$ is strictly positive on $M$. This ends the 
proof of the lemma.
\end{proof}

\begin{lemma}\label{lmSubSol2}
Let $u$ be as given in Lemma \ref{lmSubSol}.
There exists a constant $\lambda_- > 0$ such that $\phi_- = \lambda_- u$
is a positive subsolution to the system \eqref{eqLichSysMu} with $\mu \in (0, 1)$,
i.e.
\begin{subequations}\label{eqLichSysMuSub}
\begin{empheq}[left=\empheqlbrace]{align}
-\frac{4(n-1)}{n-2} \Delta \phi_- + \scal_g~\phi_- + \frac{n-1}{n} \tau^2 
\phi_-^{N-1} &\leq \frac{A^2}{(\phi_-+\mu)^{N+1}} \quad \text{on }M, \label{eqLichMuSub}\\
\frac{2(n-1)}{n-2} \partial_\nu \phi_- + H_g \phi_- &\leq  
\frac{\Theta_-}{2} 
\phi_-^{N/2} \quad \text{on }\partial M.\label{eqCondLichMuSub}
\end{empheq}
\end{subequations}
\end{lemma}

\begin{proof}
Indeed, we have
\[
 -\frac{4(n-1)}{n-2} \Delta \phi_- + \scal_g~\phi_- + \frac{n-1}{n} 
\tau^2 
\phi_-^{N-1} = \lambda (A^2)_{,\Lambda} + \lambda \frac{n-1}{n} \tau^2 u 
\left(\lambda^{N-2} u^{N-2} - 1\right),
\]
so $\phi_-$ will satisfy \eqref{eqLichMuSub} provided that
\[
\begin{aligned}
 \lambda^{N-2} u^{N-2} & \leq 1,\\
 \lambda & \leq (\lambda u + \mu)^{-N-1}
\end{aligned}
\]
(note that we used the fact that $(A^2)_{,\Lambda} \leqslant A^2$).
As $u \in W^{2, p}(M, \bR) \subset L^\infty(M, \bR)$, the first condition
is fulfilled provided that $\lambda \leq \|u\|_{L^\infty(M, \bR)}^{-1}$
and the second one if $\lambda \leq (\lambda \|u\|_{L^\infty(M, \bR)} + 
1)^{-N-1}$.
As this second condition forces $\lambda \leq 1$, we get that 
\eqref{eqLichMuSub}
is satisfied if
\[
 \lambda \leq (1 + \|u\|_{L^\infty(M, \bR)})^{-N-1}
\]
as we clearly have $(1 + \|u\|_{L^\infty(M, \bR)})^{-N-1} \leq 
\|u\|_{L^\infty(M, \bR)}^{-1}$.

Regarding the boundary condition \eqref{eqCondLichMuSub}, it is fulfilled
provided that $\dsp{-\phi_- \leq \frac{\Theta_-}{2} \phi_-^{\frac{N}{2}}}$, i.e. if
\[
 \lambda^{N/2-1} u^{N/2-1} \frac{|\Theta_-|}{2} \leq 1.
\]
Once again, this will be satisfied provided that
\[
 \lambda \leq \|u\|_{L^\infty(M, \bR)} 
\left\|\frac{|\Theta_-|}{2}\right\|_{L^\infty(\partial M, \bR)}^{\frac{2}{N-2}}.
\]

As a consequence, we have proven that, if
\begin{equation}\label{eqCondLambda}
  \lambda \leq \min \left\{(1 + \|u\|_{L^\infty(M, \bR)})^{-N-1}, 
\|u\|_{L^\infty(M, \bR)} \left\|\frac{\Theta_-}{2}\right\|_{L^\infty(\partial 
M, 
\bR)}^{\frac{2}{N-2}}\right\},
\end{equation}
the function $\phi_-$ satisfies \eqref{eqLichSysMuSub}.
\end{proof}

\begin{lemma}\label{lmMaxPrinciple}
Let $\phi_-$ be as defined in Lemma \ref{lmSubSol2}. We have
$\phi_\mu \geq \phi_-$ a.e. for all $\mu \in (0, 1)$.
\end{lemma}

\begin{proof}
The usual trick would be to subtract the equation satisfied by $\phi_\mu$
with the inequalities satisfied by $\phi_-$, multiply by $(\phi_- - 
\phi_\mu)_{0, 
}$
and integrate by parts (see e.g. \cite{HolstTsogtgerel} and also
\cite{dgh, GicquaudLichnerowicz} for
different approaches). Nevertheless, as we are only allowed to have
bounded test functions in \eqref{eqWeakForm}, we need to refine this argument
somehow. The strategy here is to take $\psi = (\phi_- - \phi_\mu)_{0, k}$
for any $k \geq 1$ and let $k$ tend to infinity. If we show that $\psi =0$, we will get that $\phi_- -\phi_\mu \leqslant 0$ a.e.\\

It follows from \eqref{eqLichSysMuSub} that $\phi_-$ satisfies this, for
all non-negative function $\psi \in W^{1, 2}(M, \bR) \cap L^\infty(M, \bR)$,

\begin{equation}\label{eqWeakFormSub}
\begin{aligned}
0 &\geq \int_M \left[\frac{4(n-1)}{n-2} \<\dd\phi_-, \dd\psi\> + 
\scal_g~\phi_-\psi 
+ \frac{n-1}{n} \tau^2 \phi_-^{N-1} \psi - \frac{A^2 
\psi}{(\phi_-+\mu)^{N+1}}\right] \dd\mu^g\\
  &\qquad + 2 \int_{\partial M} \left[H_g \phi_- - \frac{\Theta_-}{2} 
\phi_-^{N/2} \right]\psi \dd\mu^g.
\end{aligned}
\end{equation}
Hence, subtracting \eqref{eqWeakFormSub} and \eqref{eqWeakForm},
we get that, for all non-negative function $\psi \in W^{1, 2}(M, \bR) \cap 
L^\infty(M, \bR)$,
\begin{equation}\label{eqWeakFormDiff}
\begin{aligned}
0 &\geq \int_M \left[\frac{4(n-1)}{n-2} \<\dd(\phi_- - \phi_\mu), \dd\psi\> + 
\scal_g~(\phi_- - \phi_\mu) \psi \right. \\
 &\qquad + \left. \frac{n-1}{n} \tau^2 (\phi_-^{N-1} 
- 
\phi_\mu^{N-1}) \psi\right] \dd\mu^g\\
  &\qquad - \int_M A^2 \psi \left(\frac{1}{(\phi_-+\mu)^{N+1}} - 
\frac{1}{(\phi_\mu+\mu)^{N+1}}\right) \dd\mu^g\\
  &\qquad + 2 \int_{\partial M} \left[H_g (\phi_- - \phi_\mu) - 
\frac{\Theta_-}{2} (\phi_-^{N/2} - \phi_\mu^{N/2})\right]\psi \dd\mu^g.
\end{aligned}
\end{equation}

We now set, for any integer $k \geq 0$, $\psi = (\phi_- - \phi_\mu)_{0, k}
\in W^{1, 2}(M, \bR) \cap L^\infty(M, \bR)$.
We have (see \cite[Theorem 7.8]{GilbargTrudinger}), for all $x \in M$,
\[
 \dd\psi(x) =
 \left\lbrace
 \begin{aligned}
  \dd(\phi_- - \phi_\mu)(x) &\quad \text{if } 0 < (\phi_- - 
\phi_\mu)(x) < k,\\
  0 &\quad \text{elsewhere}.
 \end{aligned}
 \right.
\]
We have that the three following terms are non negative:

\begin{itemize}
\item $\dsp{\int_M \tau^2 (\phi_-^{N-1} - \phi_\mu^{N-1}) \psi \dd \mu^g}$ ,
\item $\dsp{-\int_M A^2 \psi \left(\frac{1}{(\phi_-+\mu)^{N+1}} - 
\frac{1}{(\phi_\mu+\mu)^{N+1}}\right) \dd\mu^g}$ ,
\item $\dsp{-\int_{\partial M} \left[\frac{\Theta_-}{2} (\phi_-^{N/2} - 
\phi_\mu^{N/2})\right]\psi \dd\mu^g}$ .
 \end{itemize}
As a consequence, \eqref{eqWeakFormDiff} yields
\begin{equation}\label{eqWeakFormDiff2}
\begin{aligned}
0 &\geq \int_M \left[\frac{4(n-1)}{n-2} \<\dd(\phi_- - \phi_\mu), \dd\psi\> + 
\scal_g~(\phi_- - \phi_\mu) \psi\right] \dd\mu^g\\
&\qquad + 2 \int_{\partial M} H_g (\phi_- - \phi_\mu)\psi \dd\mu^g
\end{aligned}
\end{equation}
We now let $k$ tend to $\infty$. We have
\[
(\phi_- - \phi_\mu)_{0, k} \to (\phi_- - \phi_\mu)_{0, } 
\quad\text{in } W^{1, 
2}(M, \bR).
\]
Indeed, since
\[
 0 \leq (\phi_- - \phi_\mu)_{0, k} \leq (\phi_- - \phi_\mu)_{0,}~,
\]
we have
\[
 0 \leq (\phi_- - \phi_\mu)_{0, } - (\phi_- - \phi_\mu)_{0, k} \leq 
(\phi_- - 
\phi_\mu)_{0, }
\]
and $(\phi_- - \phi_\mu)_{0, k} \to (\phi_- - \phi_\mu)_{0, }$ a.e.
Since $\phi_-$ and $\phi_\mu$ belong to $L^2(M, \bR)$, we have
$(\phi_- - \phi_\mu)_{0, } \in L^2(M, \bR)$. So we conclude from the
dominated convergence theorem that
\[
 \left\|(\phi_- - \phi_\mu)_{0, } - (\phi_- - \phi_\mu)_{0, 
k}\right\|^2_{L^2(M, 
\bR)} = \int_M \left((\phi_- - \phi_\mu)_{0, } - (\phi_- - 
\phi_\mu)_{0, 
k}\right)^2 \dd\mu^g \to_{k \to \infty} 0.
\]
This proves convergence in $L^2(M, \bR)$. Convergence of the gradient
is similar. We have
\[
 \dd((\phi_- - \phi_\mu)_{0, }) - \dd((\phi_- - \phi_\mu)_{0, k}) = 
\chi_{\left\{(\phi_- - 
\phi_\mu)_{0, }   \geq k \right\} } \dd((\phi_- - \phi_\mu)_{0, }) \to_{k \to \infty} 
0,
\]
so
\[
 \left|\dd((\phi_- - \phi_\mu)_{0, }) - \dd((\phi_- - \phi_\mu)_{0, 
k})\right|^2 
\leq |\dd((\phi_- - \phi_\mu)_{0, })|^2 \in L^1(M, \bR).
\]
From the dominated convergence theorem, we conclude again that
$\|\dd((\phi_- - \phi_\mu)_{0, }) - \dd((\phi_- - \phi_\mu)_{0, 
k})\|_{L^2(M, \bR)} 
\to 0$.
As a consequence, we can pass to the limit $k \to \infty$ in 
\eqref{eqWeakFormDiff2}:
\begin{align*}
0 &\geq \int_M \left[\frac{4(n-1)}{n-2} |\dd(\phi_- - \phi_\mu)_{0,}|^2 + 
\scal_g~(\phi_- - \phi_\mu)_{0,}^2\right] \dd\mu^g\\
 &\qquad + 2 \int_{\partial M} H_g (\phi_- - \phi_\mu)_{0, }^2 \dd\mu^g = 
Q_g((\phi_- - \phi_\mu)_{0,}).
\end{align*}
From the coercivity of $Q_g$, we conclude that $(\phi_- - \phi_\mu)_{0,} 
\equiv 
0$
a.e., namely $\phi_- \leq \phi_\mu$.
\end{proof}

A similar argument shows that $\phi_{\mu'} \leq \phi_\mu$ if $\mu' \geq 
\mu$
as $\phi_{\mu'}$ is then a weak subsolution to \eqref{eqWeakForm}.
In particular, choosing $\mu=\mu' $, this shows that $\phi_\mu$ is the only positive function
satisfying \eqref{eqWeakForm}.

\begin{lemma} \label{lmmu0}
Let $\phi \definedas \limsup_{\mu \to 0^+} \phi_\mu$. Then $\phi$
satisfies, for all $\psi \in W^{1, 2}(M, \bR) \cap L^\infty(M, \bR)$,
the weak form of the Lichnerowicz equation \eqref{eqWeakForm0}.
We have $\phi \geq \phi_-$, $\phi$ minimizes the functional $I_0$
(see \eqref{eqDefImu} for the definition) and is the only
non-negative function in $W^{1, 2}(M, \bR)$ satisfying
\eqref{eqWeakForm0}.
\end{lemma}

\begin{proof}
We first remark that
\[
 Q_g(\phi_\mu) \leq I_\mu(\phi_\mu) \leq I_\mu(1) \leq I_0 (1) .
 \]
By the coercivity of $Q_g$, we conclude that the $\phi_\mu$ are uniformly
bounded in $W^{1, 2}(M, \bR)$. As a consequence, there is a decreasing
sequence $(\mu_k)_k$ tending to zero such that $(\phi_{\mu_k})_k$ converges
to $\phi \in W^{1, 2}(M, \bR)$, weakly in
$W^{1, 2}(M, \bR)$, strongly in $L^2(M, \bR)$ and $\phi$ verifies $\phi \geq \phi_-$ a.e.. In particular, we can
pass to the limit in \eqref{eqWeakForm} and get that $\phi$ solves
\eqref{eqWeakForm0} for any
$\psi \in W^{1, 2}(M, \bR) \cap L^\infty(M, \bR)$.
 
Repeating the argument we did before, we see that $\phi$ is the only
positive function satisfying \eqref{eqWeakForm0}. All that we have to show
is that $I_0(\phi)$ is the minimum value of $I_0$ over all positive functions
in $W^{1, 2}(M, \bR)$. From the discussion regarding the weak semicontinuity
of $I_\mu$ above, we have
\begin{equation}\label{eqLimImu}
 \liminf_{k \to \infty} I_{\mu_k}(\phi_{\mu_k}) \geq I_0(\phi).
\end{equation}
Assume that there exists $\phi' \geq 0$ such that $I_0(\phi') < 
I_0(\phi)$.
We have
\[
 I_\mu(\phi') = I_0(\phi') - \frac{1}{N} \int_M A^2 \left[(\phi')^{-N} 
- 
(\phi'+\mu)^{-N}\right] \leq I_0(\phi').
\]
This contradicts \eqref{eqLimImu} since $\phi_\mu$ minimizes $I_\mu$.
\end{proof}

All that is left to prove now in the statement of Proposition \ref{propLich}
is the fact that the mapping $A \mapsto \phi$ is continuous. Assume given a
sequence $(A_k)_k$ of elements of $L^2(M, \bR)$ converging strongly to
$A \in L^2(M, \bR)$. We denote by $(\phi_k)_k$ the corresponding sequence
of solutions to the problem \eqref{eqWeakForm0}.

We first claim that $(\phi_k)_k$ is uniformly bounded from below. Let
$\Lambda$ be as in Lemma \ref{lmSubSol}. Since
$0 \leq (A_k^2)_{,\Lambda} \leq \Lambda$, and $A_k^2 \to A^2$ a.e., we
have $(A_k^2)_{, \Lambda} \to (A^2)_{, \Lambda}$ in $L^p(M, \bR)$ as a
consequence of the dominated convergence theorem.
If $u$ denotes the solution to the Robin problem \eqref{eqLinSys}
and $u_k$ that of the Robin problem \eqref{eqLinSys} where $A$ is
replaced by $A_k$, we then have $u_k \to u$ in $W^{2, p}(M, \bR)$.
In particular, the sequence $(u_k)_k$ is bounded in $L^\infty(M, \bR)$.

Choosing $\lambda_-$ sufficiently small, i.e.
\[
 \lambda_- \leq \min \left\{\inf_k (1 + \|u_k\|_{L^\infty(M, \bR)})^{-N-1}, 
\inf_k \|u_k\|_{L^\infty(M, \bR)} 
\left\|\frac{\Theta_-}{2}\right\|_{L^\infty(\partial M, \bR)}^{\frac{2}{N-2}}\right\}
\]
(see \eqref{eqCondLambda}), $\lambda_- u$ is a subsolution to \eqref{eqLichSys}
and $\lambda_- u_k$ a subsolution to \eqref{eqLichSys} with $A$ replaced by
$A_k$. As $u_k \to u$ in $L^\infty(M, \bR)$ and $u_k$ and $u$ are all
bounded from below by a positive constant, there exists an $\varepsilon > 0$
such that $u_k, u \geq \varepsilon$. As $\lambda_- u$ (resp. $\lambda_- u_k$)
is a subsolution to \eqref{eqLichSys} (resp. to \eqref{eqLichSys} with
$A$ replaced by $A_k$, see Lemma \ref{lmSubSol2}), we have
$u_k, u \geq \eta \definedas \lambda_- \varepsilon > 0$.

We now proceed as in the proof of Lemma \ref{lmSubSol2}. We subtract the
weak forms \eqref{eqWeakForm} of the system satisfied by $\phi$ and 
$\phi_k$:
\[
\begin{aligned}
0 &= \int_M \left[\frac{4(n-1)}{n-2} \<\dd(\phi - \phi_k), \dd\psi\> + 
\scal_g~(\phi 
- \phi_k) \psi + \frac{n-1}{n} \tau^2 (\phi^{N-1} - \phi_k^{N-1}) 
\psi\right] 
\dd\mu^g\\
  &\qquad - \int_M \psi \left(\frac{A^2}{\phi^{N+1}} - 
\frac{A_k^2}{\phi_k^{N+1}}\right) \dd\mu^g\\
  &\qquad + 2 \int_{\partial M} \left[H_g (\phi - \phi_k) - 
\frac{\Theta_-}{2} 
(\phi^{N/2} - \phi_k^{N/2})\right]\psi \dd\mu^g.
\end{aligned}
\]
We choose $\psi = (\phi - \phi_k)_{-K, K}$ for some large
$K > 0$. Note that if $f$ is an increasing function, we have for all $x$ and $y$, $(x-y)(f(x)-f(y))\geq 0$. From this remark, we see that
\[
 0 \leq \frac{n-1}{n} \int_M \tau^2 (\phi^{N-1} - \phi_k^{N-1}) \psi 
\dd\mu^g - 
\int_{\partial M} \frac{\Theta_-}{2} (\phi^{N/2} - \phi_k^{N/2})\psi 
\dd\mu^g,
\]
so we get
\begin{equation}\label{eqWeakFormDiff3}
\begin{aligned}
\int_M \psi \left(\frac{A^2}{\phi^{N+1}} - 
\frac{A_k^2}{\phi_k^{N+1}}\right) 
\dd\mu^g
  &\geq \int_M \left[\frac{4(n-1)}{n-2} \<\dd(\phi - \phi_k), \dd\psi\> + 
\scal_g~(\phi - \phi_k) \psi\right] \dd\mu^g\\
  &\qquad + 2 \int_{\partial M} H_g (\phi - \phi_k)\psi \dd\mu^g.
\end{aligned}
\end{equation}
We rework the left-hand side as follows. 
Assume first that $\phi_k \leq \phi$ at some point of $M$, then
\[
 \psi \left(\frac{A^2}{\phi^{N+1}} - \frac{A_k^2}{\phi_k^{N+1}}\right)
 = \frac{A^2 - A_k^2}{\phi^{N+1}} \psi + A_k^2 \left(\frac{1}{\phi^{N+1}} 
- 
\frac{1}{\phi_k^{N+1}}\right) \psi.
\]
As $x \mapsto x^{-N-1}$ is decreasing and $x \mapsto (x)_{-K, K}$
increasing, we get that
\[
 A_k^2 \left(\frac{1}{\phi^{N+1}} - \frac{1}{\phi_k^{N+1}}\right) \psi 
\leq 0.
\]
Since, $\phi_k \geq 0$, we have $0 \leq \psi \leq \phi$. Hence,
\[
 \psi \left(\frac{A^2}{\phi^{N+1}} - \frac{A_k^2}{\phi_k^{N+1}}\right)
 \leq \frac{A^2 - A_k^2}{\phi^{N+1}} \psi \leq \frac{|A^2 - 
A_k^2|}{\phi^N} \leq 
\eta^{-N} |A^2 - A_k^2|.
\]
Similar calculations show that the previous inequality also holds when
$\phi_k \geq \phi$. Using \eqref{eqWeakFormDiff3}, we have thus proven
\[
\begin{aligned}
\eta^{-N} \int_M |A^2 - A_k^2| \dd\mu^g
  &\geq \int_M \left[\frac{4(n-1)}{n-2} \<\dd(\phi - \phi_k), \dd\psi\> + 
\scal_g~(\phi - \phi_k) \psi\right] \dd\mu^g\\
  &\qquad + 2 \int_{\partial M} H_g (\phi - \phi_k)\psi \dd\mu^g.
\end{aligned}
\]
Letting $K$ tend to $\infty$ and arguing as in the proof of Lemma
\ref{lmSubSol2}, we conclude that
\[
 Q_g(\phi - \phi_k) \leq \eta^{-N} \int_M |A^2 - A_k^2| \dd\mu^g.
\]
From the coercivity of $Q_g$, we obtain that $\phi - \phi_k \to 0$ in
$W^{1, 2}(M, \bR)$.

This ends the proof of Proposition \ref{propLich}. 

\subsection{Regularity theory for the solutions of the Lichnerowicz equation}

We now get an improved
estimate for the solution $\phi$ :
\begin{proposition}\label{propImprovedRegL2}
The solution $\phi$ constructed in Proposition \ref{propLich} satisfies
$\phi^{\fraccheloue} \in W^{1, 2}(M, \bR)$ and
$\left\|\phi^{\fraccheloue}\right\|_{W^{1, 2}(M, \bR)} \lesssim \|A\|_{L^2(M, 
\bR)}$, where $U_1 \lesssim U_2$ means that $U_1 \leq c U_2$ for some constant $c$ independant of $\phi$ and $A$. 
\end{proposition}

\begin{proof}
Calculations similar to the one done to get Equation \eqref{eqConfTransQ}
yield, if $\phibar = u_0^{-1} \phi$. We remind the reader that $g_0 =u_0^{N-2} g$ where $u_0$ was defined in Proposition \ref{propEigenFunction}.
\[
 I_0(\phi) = Q_{g_0}(\phibar) + \frac{1}{N} \int_M \left[\tau^2 
\phibar^N 
\dd\mu^{g_0} + A^2 u_0^{-2N} \phibar^{-N}\right] \dd\mu^{g_0} - \frac{2}{N+2} 
\int_{\partial M} \Theta_- \phibar^{\fraccheloue} \dd\mu^{g_0}.
\]
As a consequence, $\phi$ solves \eqref{eqWeakForm0} if and only if
$\phibar$ satisfies, for any $\psi \in W^{1, 2}(M, \bR)\cap L^\infty (M,\bR)$,
\begin{equation}\label{eqWeakFormConf}
\begin{aligned}
0 &= \int_M \left[\frac{4(n-1)}{n-2} \<\dd\phibar, \dd\psi\>_{g_0} + \lambda_0(M, 
g) 
u_0^{2-N} \phibar\psi + \frac{n-1}{n} \tau^2 \phibar^{N-1} \psi - 
\frac{A^2 
u_0^{-2N} \psi}{\phibar^{N+1}}\right] \dd\mu^{g_0}\\
  &\qquad \int_{\partial M} (-\Theta_-) \phibar^{N/2} \psi \dd\mu^{g_0}.
\end{aligned}
\end{equation}
Now, the idea is to set $\psi = \psi_k = (\phibar^{N+1})_{, k}$ in 
\eqref{eqWeakFormConf}
for $k \geq 0$ and let $k$ tend to infinity. Since we have
\[
 \dd\psi = \left\lbrace\begin{aligned} (N+1) \phibar^N \dd\phibar & \quad\text{if } 
\phibar^{N+1} < 
k,\\ 0 & \quad \text{otherwise,}\end{aligned}\right.
\]
we see that the sequence of functions 
\[
 \<\dd\phibar, \dd\psi_k\>_{g_0} = \frac{N+1}{\dsp{\left(\frac{N}{2}+1\right)^2}} 
\left\lbrace\begin{aligned} |\dd(\phibar^{\fraccheloue})|^2 & \quad\text{if } 
\phibar^{N+1} < k,\\ 0 & \quad \text{otherwise.}\end{aligned}\right.
\]
is increasing and converges a.e. to
$\displaystyle \frac{N+1}{\dsp{\left(\frac{N}{2}+1\right)^2}} 
|\dd(\phibar^{\fraccheloue})|^2$.
By the monotone convergence theorem, we conclude that
\begin{align*}
 &\int_M \left[\frac{4(n-1)}{n-2} \<\dd\phibar, \dd\psi_k\>_{g_0} + \lambda_0(M, 
g) 
u_0^{2-N} \phibar\psi_k + \frac{n-1}{n} \tau^2 \phibar^{N-1} 
\psi_k\right] 
\dd\mu^{g_0}\\
 &\qquad \to_{k \to \infty} \int_M \left[\frac{3n-2}{n-1} 
|\dd(\phibar^{\fraccheloue})|^2_{g_0} + \lambda_0(M, g) \phibar^{N+2} + 
\frac{n-1}{n} 
\tau^2 \phibar^{2N}\right] \dd\mu^{g_0}
\end{align*}
Similarly,
\[
 \int_{\partial M} (-\Theta_-) \phibar^{N/2} \psi_k \dd\mu^{g_0} \to_{k \to 
\infty} \int_{\partial M} (-\Theta_-) \phibar^{\frac{3N}{2}+2} \dd\mu^{g_0}
\]
and by dominated convergence
\[
 \int_M \frac{A^2 u_0^{-2N} \psi_k}{\phibar^{N+1}} \dd\mu^{g_0} \to_{k \to 
\infty} 
\int_M A^2 u_0^{-2N} \dd\mu^{g_0}.
\]
Thus, passing to the limit in \eqref{eqWeakFormConf} we have shown that
\begin{align*}
 0
 &= \int_M \left[\frac{3n-2}{n-1} |\dd(\phibar^{\fraccheloue})|^2_{g_0} + \lambda_0(M, 
g) 
\phibar^{N+2} + \frac{n-1}{n} \tau^2 \phibar^{2N}\right] \dd\mu^{g_0}\\
 &\qquad + \int_{\partial M} (-\Theta_-) \phibar^{\frac{3N}{2}+2} \dd\mu^{g_0} - 
\int_M 
A^2 u_0^{-2N} \dd\mu^{g_0}.
\end{align*}
In particular,
\[
 \int_M \left[\frac{3n-2}{n-1} |\dd(\phibar^{\fraccheloue})|^2_{g_0} + \lambda_0(M, 
g) 
\phibar^{N+2}\right]\dd\mu^{g_0} \leq \int_M A^2 u_0^{-2N} \dd\mu^{g_0}.
\]
As $u_0, u_0^{-1} \in W^{2, p}(M, \bR) \subset L^\infty(M, \bR)$, the multiplication by $u_0$ is an automorphism of $W^{1,2}(M,\bR)$, so we conclude that
\[
 \|\phi^{\frac{N}{2}+1}\|_{W^{1, 2}(M, \bR)} \lesssim \|A\|_{L^2(M, \bR)}.
\]

Regularity for $\phi^N$ can now be obtained as follows. We set
$\psi = \phi^{\frac{N}{2}+1}$ and $\dsp{a = \frac{N}{\dsp{\fraccheloue}}}$ so $\phi^N = \psi^a$.
and $\dd\phi^N = a \psi^{a-1} \dd\psi$. From the Sobolev embedding theorem, we know 
that $\psi \in L^N(M, \bR)$ so $\psi^{a-1} \in L^{N/(a-1)}(M, \bR)$. Thus, from 
H\"older's inequality, we get
\[
 \left\|\dd\phi^N\right\|_{L^{r_0}(M, \bR)}
  \leq \|\psi^{a-1}\|_{L^{N/(a-1)}(M, \bR)} \|\dd\psi\|_{L^2(M, \bR)}
\]
provided that $\dsp{\frac{1}{r_0} = \frac{a-1}{N} + \frac{1}{2}}$. Simple calculations 
show that $\dsp{r_0 = \frac{2n(n-1)}{n^2-2}}$ as claimed. Chasing out estimates in 
terms of $\|A\|_{L^2(M, \bR)}$, we conclude that 
$$\left\|\phi^N\right\|_{W^{1, 
r_0}(M, \bR)} \lesssim \|A\|^{1 - \frac{1}{n}}_{L^2(M, \bR)}.$$
\end{proof}

The estimate of Proposition \ref{propImprovedRegL2} allows us to prove continuity
and compactness of the mapping $A \mapsto \phi^N$. This will be useful in
Section \ref{secCoupled} where we apply the Schauder fixed point theorem:

\begin{proposition}\label{propContinuity}
Let $\Phi: L^2(M, \bR) \to L^a(M, \bR) \times L^b(\partial M, \bR)$ be the
mapping associating to an $A$ in $L^2(M, \bR)$ the pair
$(\phi^N, \phi^N \vert_{\partial M})$, where $\phi$ denotes the solution to the
weak formulation of the Lichnerowicz equation (see Proposition \ref{propLich}).
Then $\Phi$ is well-defined, continuous and compact provided
\[
 a \in \left[1, \frac{N}{2}+1\right)\quad\text{and}\quad
 b \in \left[1, \frac{1}{N}\left(\frac{N}{2}+1\right)^2\right).
\]
\end{proposition}

\begin{proof}
We first prove that the mapping $A \mapsto \phi^N$ is continuous and compact
from $L^2(M, \bR)$ to $L^1(M, \bR)$ (i.e. when $a = 1$) and deduce the result
for larger values of $a$ by interpolation.

From Proposition \ref{propLich}, we know that the mapping $A \mapsto \phi$ is
continuous as a map from $L^2(M, \bR)$ to $W^{1, 2}(M, \bR)$. From the Sobolev
embedding theorem, we get that $A \mapsto \phi$ is continuous as a map to
$L^N(M, \bR)$. Now, from the mean value theorem, if $\phi_1$ and $\phi_2$ are
two positive functions in $L^N(M, \bR)$, there exists a function
$\psi \in [\phi_1, \phi_2]$ such that, for all $x \in M$
\[
 \phi_1^N(x) - \phi_2^N(x) = N (\phi_1(x) - \phi_2(x)) \psi(x)^{N-1}.
\]
In particular, we have
\[
 0 \leq \psi^{N-1}(x) \leq \max\{\phi_1^{N-1}(x), \phi^{N-1}_2(x)\} \leq \phi_1^{N-1}(x) + \phi_2^{N-1}(x).
\]
Thus,
\begin{align*}
 \left\|\phi_1^N - \phi_2^N\right\|_{L^1(M, \bR)}
 &= \int_M \left|\phi_1^N - \phi_2^N\right| \dd\mu^g\\
 &= N \int_M |\phi_1(x) - \phi_2(x)| \psi^{N-1} \dd\mu^g\\
 &\leq N \int_M |\phi_1 - \phi_2| \left(\phi_1^{N-1} + \phi_2^{N-1}\right) \dd\mu^g\\
 &\leq N \|\phi_1 - \phi_2\|_{L^N(M, \bR)} \left(\|\phi_1\|^{N-1}_{L^N(M, \bR)} + \|\phi_2\|^{N-1}_{L^N(M, \bR)}\right),
\end{align*}
where the last line follows from H\"older's inequality.
This proves that the mapping $\phi \mapsto \phi^N$ is locally Lipschitz and, hence,
continuous from $L^N(M, \bR)$ to $L^1(M, \bR)$. Continuity of $A \mapsto \phi^N$
follows by composition.

To prove compactness, we utilize Proposition \ref{propImprovedRegL2}: if we are
given a bounded sequence $(A_k)_k$ of elements of $L^2(M, \bR)$, the corresponding
sequence of solutions $(\phi_k)_k$ are such that
$\|\phi_k^{\fraccheloue}\|_{W^{1, 2}(M, \bR)}$ is bounded. Set
$\theta_k \definedas \phi_k^{\fraccheloue}$. From Rellich's theorem, the embedding
$W^{1, 2}(M, \bR) \hookrightarrow L^2(M, \bR)$ is compact so there is a
subsequence of $(\theta_k)_k$ that converges in $L^2(M, \bR)$ to some
$\theta_\infty$. Considering that subsequence, there is no loss of
generality assuming that the sequence $(\theta_k)_k$ does converges to
$\theta_\infty$. As there is a subsequence of $(\theta_k)_k$ that weakly
converges in $W^{1, 2}(M, \bR)$, we have $\theta_\infty \in W^{1, 2}(M, \bR)$
(see e.g. the proof of Lemma \ref{lmEscobar}).

Let $\dsp{\alpha \definedas \frac{N}{\dsp{\fraccheloue}}=\frac{n}{n-1}}$ so $\phi_k^N = \theta_k^\alpha$. By
repeating the above argument with the mean value theorem, we have
\begin{align*}
 \left\|\phi_k^N - \phi_\infty^N\right\|_{L^1(M, \bR)}
 & = \int_M |\phi_k^N - \phi_\infty^N| \dd\mu^g\\
 & \leq \alpha \int_M |\theta_k - \theta_\infty| \left(\theta_k^{\alpha-1} + \theta_\infty^{\alpha-1}\right) \dd\mu^g\\
 & \leq \alpha \vol(M, g)^\beta \|\theta_k - \theta_\infty\|_{L^2(M, \bR)}\\ &\qquad \times \left(\left\|\theta_k^{\alpha-1}\right\|_{L^{N/(\alpha-1)}(M, \bR)} + \left\|\theta_\infty^{\alpha-1}\right\|_{L^{N/(\alpha-1)}(M, \bR)}\right)\\
 & \leq \alpha \vol(M, g)^\beta \|\theta_k - \theta_\infty\|_{L^2(M, \bR)} \left(\left\|\theta_k\right\|_{L^N(M, \bR)}^{\alpha-1} + \left\|\theta_\infty\right\|_{L^N(M, \bR)}^{\alpha-1}\right),
\end{align*}
where $\dsp{\beta = \frac{n^2 - 2n + 2}{2n(n-1)}}$.
As the norms $\left\|\theta_k\right\|_{L^N(M, \bR)}$ are uniformly bounded, we
conclude that $\dsp{\phi_k^N \to \phi_\infty^N}$ in $L^1(M, \bR)$ proving that the
mapping $A \mapsto \phi^N$ is compact as a mapping from $L^2(M, \bR)$ to $L^1(M, \bR)$.

Continuity and compactness of the mapping $A \mapsto \phi^N$ from $L^2(M, \bR)$
to $L^a(M, \bR)$, $\displaystyle a \in \left[1, \frac{N}{2}+1\right)$ is then
obtained easily from the interpolation inequality (see \cite[Inequality ($7.9$)]{GilbargTrudinger}) Let $\lambda \in [0, 1)$ be
such that
\[
 \frac{1}{a} = \frac{1-\lambda}{1} + {\lambda}{\frac{N}{2}+1},
\]
then we have, for any $\zeta \in L^{\fraccheloue}(M, \bR)$,
\[
 \|\zeta\|_{L^a(M, \bR)} \leq \|\zeta\|_{L^1(M, \bR)}^{1-\lambda} \|\zeta\|_{L^{\fraccheloue}(M, \bR)}^{\lambda}.
\]
Indeed, this inequality shows that if we are given a bounded sequence of
functions $(\phi_k^N)_k$ in $L^{\fraccheloue}(M, \bR)$ that converges in $L^1(M, \bR)$, it also converges in $L^a(M, \bR)$ for we can write
\[
\left\|\varphi_k^N - \varphi_{k'}^N\right\|_{L^a(M, \Rb)}
 \leqslant \underbrace{\|\varphi_k^N - \varphi_{k'}^N\|_{L^1(M, \Rb)}^{1-\lambda}}_{\text{tends to zero when } k, k' \to \infty} \quad
\underbrace{\|\varphi_k^N - \varphi_{k'}^N\|_{L^{\fraccheloue}(M, \Rb)}^{\lambda}}_{\text{bounded indepedently from $k$ and $k'$}}.
\]

Continuity and compactness of $A \mapsto \phi^N \vert_{\partial M}$ is proven by
the same methods. Details are left to the reader.
\end{proof}

We conclude this section by proving that, if $A$ enjoys better regularity than 
belonging to $L^2(M, \bR)$, $\phi$ is a strong solution to the Lichnerowicz 
equation that satisfies enhanced regularity.

\begin{lemma}\label{lmLich2}
Assume that $A \in L^{2q}(M, \bR)$ for some $q \in [1, p]$. Then the solution 
$\phi$ constructed in Proposition \ref{propLich} satisfies $\phi^N \in L^t(M, \Rb)$ with
 $$
 \left\lbrace
 \begin{aligned}
   t =  \frac{2(n-1)q}{n-2q} & \qquad\text{if } q < \frac{n}{2},\\
   t \in [1; \infty[\text{ arbitrary } & \qquad\text{otherwise}.
 \end{aligned}
 \right.
 $$
 In particular, $\phi^N \in W^{1, r}(M, \Rb)$ with 
\[
r = \left\lbrace
\begin{aligned}
\frac{n^2 - 2q}{2n(n-1)q} & \qquad \textrm{if } q < \frac{n^2}{n^2 - n + 2},\\
2 & \qquad \textrm{otherwise.}
\end{aligned}
\right.
\]
\end{lemma}

\begin{proof}
Our first task is to get an estimate for $\phi^N$ in some higher $L^r$-spaces.
We reproduce the proof of Proposition \ref{propImprovedRegL2} except that we
choose for $\psi$ higher powers of $\phibar$ in Equation
\eqref{eqWeakFormConf}. Let $\alpha \in \bR_+$ be some number to be chosen
later. We set $\psi = \psi_k = (\phibar)_{,k}^{N+1+2\alpha}$. We have
\[
 \dd\psi = \left\lbrace\begin{aligned}
  (N+1+2\alpha) \phibar^{N+2\alpha} \dd\phibar & \quad\text{if } \phibar < k,\\
  0 & \quad \text{otherwise.}
 \end{aligned}\right.
\]
So,
\begin{align*}
 \<\dd\phibar, \dd\psi_k\>_{g_0}
 &= \frac{N+1+2\alpha}{\dsp{\left(\frac{N}{2}+1+\alpha\right)^2}} 
\left\lbrace\begin{aligned} |\dd(\phibar^{\fraccheloue+\alpha})|^2 & \quad\text{if } 
\phibar < k,\\ 0 & \quad \text{otherwise}\end{aligned}\right.\\
 &= \frac{N+1+2\alpha}{\dsp{\left(\frac{N}{2}+1+\alpha\right)^2}} \left|\dd\left((\phibar)_{,k}^{\fraccheloue+\alpha}\right)\right|^2,
\end{align*}
Equation \eqref{eqWeakFormConf} then leads to the following inequality:
\begin{equation}\label{eqWeakFormConf2}
\begin{aligned}
&\int_M \left[\frac{4(n-1)}{n-2} \frac{N+1+2\alpha}{\dsp{\left(\frac{N}{2}+1+\alpha\right)^2}}  \left|\dd(\phibar)_{,k}^{\fraccheloue+\alpha}\right|_{g_0}^2
+ \lambda_0(M, g) u_0^{2-N} \left((\phibar)_{,k}^{\fraccheloue+\alpha}\right)^2\right] \dd\mu^{g_0}\\
&\qquad \leq \int_M A^2 u_0^{-2N} (\phibar)_{,k}^{2 \alpha} \dd\mu^{g_0}.
\end{aligned}
\end{equation}
From the Sobolev embedding theorem together with H\"older's inequality, we
infer
\begin{align*}
 \left\|(\phibar)_{,k}^{\fraccheloue+\alpha}\right\|_{L^N(M, \bR)}^2
 &\lesssim \int_M A^2 u_0^{-2N} (\phibar)_{,k}^{2 \alpha} \dd\mu^{g_0}\\
 &\lesssim \|A\|^2_{L^{2q}(M, \bR)} \|(\phibar)_{,k}^{2 \alpha}\|_{L^s(M, \bR)},
\end{align*}
where $r$ is such that $\dsp{\frac{1}{q} + \frac{1}{s} = 1}$. Thus, rearranging powers,
\[
 \left\|(\phibar)_{,k}^N\right\|_{L^{\fraccheloue+\alpha}(M, \bR)}^{\frac{2}{N} (\fraccheloue+\alpha)}
 \lesssim \|A\|^2_{L^{2q}(M, \bR)} \|(\phibar)_{,k}^N\|_{L^{\frac{2\alpha s}{N}}(M, \bR)}^{\frac{2\alpha}{N}}.
\]
The best we can hope for is that the Lebesgue norms of $(\phibar)_{,k}^N$
appearing on both sides coincide. Straightforward calculations show that we get
\[
 \alpha = \left(\frac{N}{2}+1\right) \frac{q-1}{\dsp{1-\frac{2q}{n}}}.
\]
This choice is valid as long as $\dsp{q < \frac{n}{2}}$ but $\alpha$ becomes negative
as soon as $\dsp{q > \frac{n}{2}}$. Assume first that $\dsp{q < \frac{n}{2}}$. Then we get, 
by simplifying the previous inequality, that
\[
 \left\|(\phibar)_{,k}^N\right\|_{L^{\fraccheloue+\alpha}(M, \bR)}^{1 + \frac{2}{N}}
 \lesssim \|A\|^2_{L^{2q}(M, \bR)}
\]
and, letting $k$ tend to infinity, we conclude that
\[
 \left\|\phi^N\right\|_{L^t(M, \bR)}^{1 + \frac{2}{N}}
 \lesssim \|A\|^2_{L^{2q}(M, \bR)},
\]
where $\dsp{t = \fraccheloue+\alpha = \frac{2(n-1)q}{n-2q} > q}$. If $\dsp{q \geq \frac{n}{2}}$, we have, for any $\alpha$,
that 
$$\dsp{\frac{N}{2} + 1 + \alpha \geq \frac{2 s}{N} \alpha}.$$ 
In particular,
\[
 \|(\phibar)_{,k}^N\|_{L^{\frac{2\alpha s}{N}}(M, \bR)} \lesssim \left\|(\phibar)_{,k}^N\right\|_{L^{\fraccheloue+\alpha}(M, \bR)}.
\]
So, by a similar reasoning, we get that $\phi^N \in L^t(M, \bR)$ for any $t \in
[1, \infty)$. Sobolev regularity for $\phi^N$ can then be obtained by means
similar to the one in Proposition \ref{propImprovedRegL2}. Note that, if $\dsp{q \geq
\frac{n^2}{n^2 - n + 2}}$, we can choose $\dsp{\alpha = \frac{N}{2}-1}$ in the previous
argument which directly leads to the fact that $\phi^N \in W^{1, 2}(M, \bR)$.
\end{proof}

Our next claim is that weak solutions to the Lichnerowicz equation are actually strong solutions if $A$ is regular enough.

\begin{prop}\label{propLichStrong}
Assume that $A \in L^{2q}(M, \bR)$ for some $q \in ]1,p]$. Then the weak solution $\phi$ constructed un Proposition \ref{propLich} is a strong solution and belongs to $W^{2, q}(M, \bR)$. We have $\phi^N \in W^{1, r}(M, \Rb)$ with 
$$
r = \left\{
\begin{aligned}
 \frac{8n(n-1)q}{7n^2 - 2n - (6n+4)q} &\qquad\text{if } q < \frac{n}{2},\\
 \frac{nq}{n-q} &\qquad\text{if } \frac{n}{2} < q < n,\\
 \infty &\qquad\text{if } q > n.
\end{aligned}
\right.
$$
\end{prop}

\begin{proof}
We first show that $\phi$ is a strong solution of the system \eqref{eqLichSys}. By using Proposition \ref{lmLich2} and the fact that $\phi$ is bounded below (Lemma \ref{lmmu0}), we see that
$$
-\frac{n-1}{n} \tau^2 \varphi^{N-1} + \frac{A^2}{\varphi^{N+1}} \in L^q(M, \Rb).
$$
Similarly, $\varphi^{N/2} \in W^{1, 2}(M, \Rb)$, so $\Theta_- \varphi^{N/2} \in W^{1, 2}(M, \Rb)$. Set $q_0 = \min \{2, q\}$, $q_0 >1$. There exists a unique function $\varphi_0 \in W^{2, q_0}(M, \Rb)$ satisfying
\begin{equation}\label{eqPhi0}
\left\lbrace
\begin{aligned}
-\frac{4(n-1)}{n-2} \Delta \varphi_0 + \Scal_{g}\varphi_0
&= - \frac{n-1}{n} \tau^2 \varphi^{N-1} + A^2 \varphi^{-N-1} &\qquad\text{on }M,\\
\frac{2(n-1)}{n-2} \partial_\nu \varphi_0 + H_g \varphi_0
&= \frac{\Theta_-}{2}  \varphi^{\frac{N}{2}} &\qquad\text{on } \partial M.
\end{aligned}
\right.
\end{equation}
Notice this is \eqref{eqLichSys} in which we have replaced the left-hand members by the corresponding expression for $\varphi_0$. The existence of $\varphi_0$ is granted by the fact that $\mathcal{E}(M, g) > 0$ (see, for example, the proof of Lemma \ref{lmSubSol}).

We will now show that $\varphi_0 \equiv \varphi$, which will allow us to conclude that $\varphi \in W^{2, q_0}(M, \Rb)$. The difficulty here comes from the fact that $q$ can be as close to $1$ as we want so we do not necessarily have $W^{2, q}(M, \Rb) \hookrightarrow W^{1, 2}(M, \Rb)$. To overcome it, the idea here is to see $\varphi_0$ and $\varphi$ as very weak solutions of the same problem. Remark that if $\psi \in W^{2, p}(M, \Rb)$, we get $\psi \in W^{1, 2}(M, \Rb) \cap L^\infty(M, \Rb)$ so $\psi$ is a valid test function in the weak formulation of \eqref{eqWeakForm0} of the Lichnerowicz equation. Likewise, multiplying \eqref{eqPhi0} by $\psi$ and integrating by parts, we can write
$$
\begin{aligned}
0 &= \int_M \left[\frac{4(n-1)}{n-2} \langle \dd\varphi_0, \dd\psi\rangle + \Scal_{g} \varphi_0 \psi 
+ 
\frac{n-1}{n} \tau^2 \varphi^{N-1} \psi - \frac{A^2 \psi}{\varphi^{N+1}}\right] \dd\mu^g\\
&\qquad + 2 \int_{\partial M} \left[H_g \varphi_0 - \frac{\Theta_-}{2} 
\varphi^{\frac{N}{2}} \right]\psi \dd\mu^g.
\end{aligned}
$$
By substracting that last calculation with \eqref{eqWeakForm0}, we then get
$$
\begin{aligned}
0 &= \int_M \left[\frac{4(n-1)}{n-2} \langle \dd(\varphi - \varphi_0), \dd\psi\rangle + \Scal_{g} (\varphi - \varphi_0) \psi\right] \dd\mu^g\\
  &\qquad + 2 \int_{\partial M} H_g (\varphi - \varphi_0) \psi \dd\mu^g.
\end{aligned}
$$
Once again integrating by parts gives us
\begin{equation}\label{eqWeakFormDiff4}
\begin{aligned}
0 &= \int_M \left[-\frac{4(n-1)}{n-2} \Delta \psi + \Scal_{g} \psi\right](\varphi - \varphi_0) \dd\mu^g\\
&\qquad + 2 \int_{\partial M}\left[\frac{2(n-1)}{n-2} \partial_\nu \psi + H_g  \psi\right] (\varphi - \varphi_0) \dd\mu^g.
\end{aligned}
\end{equation}
We will now choose $\psi$ to conclude that $\varphi-\varphi_0 \equiv 0$. Pick an arbitrary constant $K > 0$. We have $(\varphi - \varphi_0)_{-K, K} \in L^p(M, \Rb)$ so the solution $\psi$ of the problem 
\begin{equation}\label{eqPsi}
\left\lbrace
\begin{aligned}
-\frac{4(n-1)}{n-2} \Delta \psi + \Scal_{g} \psi
&= (\varphi-\varphi_0)_{-K, K} &\qquad\text{on }M\\
\frac{2(n-1)}{n-2} \partial_\nu \psi + H_g \psi
&= 0 &\qquad\text{on } \partial M
\end{aligned}
\right.
\end{equation}
is in $W^{2, p}(M, \Rb)$. Inserting it in \eqref{eqWeakFormDiff4}, we get
$$
\int_M (\varphi - \varphi_0) (\varphi - \varphi_0)_{-K, K} \dd\mu^g = 0
$$
then, by letting $K$ tend to infinity, the monotone convergence theorem gives
$$
\int_M (\varphi - \varphi_0)^2 \dd\mu^g = 0,
$$
which shows $\varphi \equiv \varphi_0$. So $\varphi \in W^{2, q_0}(M, \Rb)$ is a strong solution of problem \eqref{eqLichSys}. Let us now demonstrate that $\varphi \in W^{2, q}(M, \Rb)$, i.e. that $\varphi$ has the optimum regularity for this problem. We have already dealt with $q \leqslant 2$ because, in this case, $q = q_0$. Assume then $q > 2$ and let us assume as well that we have already shown that $\varphi \in W^{2, q_i}(M, \Rb)$ for a certain $q_i \geqslant 2$. We know that
$$
-\frac{n-1}{n} \tau^2 \varphi^{N-1} + \frac{A^2}{\varphi^{N+1}} \in L^q(M, \Rb).
$$
The only term we have to care about, regularity-wise, is then the boundary term $\dsp{\frac{\Theta_-}{2} \varphi^{N/2}}$. However, if $\dsp{q < \frac{n}{2}}$, we get that
$$
\dd \left(\frac{\Theta_-}{2} \varphi^{N/2}\right)
= \underbrace{\varphi^{N/2}}_{\in L^{\frac{4(n-1)q}{n-2q}}} \underbrace{\frac{\dd\Theta_-}{2}}_{\in L^{2p}} + \frac{N \Theta_-}{4} \underbrace{\varphi^{\frac{N}{2}-1}}_{\in L^{\frac{(n-1)(n-2)q}{n-2q}}} \underbrace{\dd\varphi}_{\in L^{q_i'}},
$$
with $q_i'$ defined by $\dsp{\frac{1}{q_i'} = \frac{1}{q_i} - \frac{1}{n}}$. We then have $\dsp{\dd \left(\frac{\Theta_-}{2} \varphi^{N/2}\right) \in L^{s}(M, \Rb)}$
with
$$
\frac{1}{s} = \max \left\{
\frac{n-2q}{4(n-1)q} + \frac{1}{2p}, \frac{1}{q_i} - \frac{(n^2 - n + 2) q - n^2}{n(n-1)(n-2)q}\right\} \leqslant \max\left\{\frac{1}{q}, \frac{1}{q_i} - c\right\},
$$
with $\dsp{c = \frac{(n^2 - n + 2) q - n^2}{n(n-1)(n-2)q} > 0}$ since we have assumed $q > 2$. Using elliptic regularity for the system \eqref{eqPhi0}, we then have $\varphi \equiv \varphi_0 \in W^{2, q_{i+1}}(M, \Rb)$ with $q_{i+1} = \min\{q, s\}$. We observe that, for a finite $i$, we have $q_{i+1} = q$, i.e. $\varphi \in W^{2, q}(M, \Rb)$. The case $\dsp{q \geqslant \frac{n}{2}}$ is similar, we will omit the proof.

Finally, we dive into the regularity of $\varphi^N$ and see how the elliptic estimates we just got will improve the result of Proposition \ref{lmLich2}. By this proposition, we have a  stronger Lebesgue estimate on $\varphi$ than the one given by elliptic regularity. Notice that, if $\dsp{q > \frac{n}{2}}$, $\varphi \in L^\infty(M, \Rb)$ and $\dd\varphi^N \in W^{1, q'}(M, \Rb)$ with $q'$ such that $\dsp{\frac{1}{q'} = \frac{1}{q} - \frac{1}{n}}$.
For smaller values of $q$, we will use Gagliardo-Nirenberg's inequality (see for example \cite[Theorem 12.87]{LeoniSobolev}).
Since $\varphi \in W^{2, q}(M, \Rb) \cap L^{Nt}(M, \Rb)$ where $t$ is defined in Lemma \ref{lmLich2}, we have $\varphi \in W^{1, s}(M, \Rb)$ for any $s$ such that
$$
\frac{1}{s} = \frac{1}{n} + \lambda \left(\frac{1}{q} - \frac{2}{n}\right) + \frac{1-\lambda}{Nt}.
$$
with $\dsp{\lambda \in \left[\frac{1}{2}; 1\right]}$. Notice that, for $\lambda = 1$, this inequality gives back the classic Sobolev embedding. Using that $\dd\varphi^N = N \varphi^{N-1} \dd\varphi$, we then get $\varphi^N \in W^{1, r}(M, \Rb)$ with
\begin{align*}
\frac{1}{r} &= \frac{N-1}{Nt} + \frac{1}{n} + \lambda \left(\frac{1}{q} - \frac{2}{n}\right) + \frac{1-\lambda}{Nt}\\
&= \frac{N-\lambda}{Nt} + \frac{1}{n} + \lambda \left(\frac{1}{q} - \frac{2}{n}\right).
\end{align*}
The optimum value of $\lambda$ (that gives the greatest $r$) depends then on the sign of the coefficient of $\lambda$, namely $\dsp{\frac{1}{q} - \frac{2}{n} - \frac{1}{Nt}}$ :
\begin{itemize}
 \item If $\dsp{\frac{1}{q} - \frac{2}{n} - \frac{1}{Nt}} < 0$, $\lambda = 1$ gives the best $r$. This condition is equivalent to $\dsp{q > \frac{n}{2}}$.
 \item Otherwise, $\dsp{\lambda = \frac{1}{2}}$ provides the best $r$. We then get
 $$
 \frac{1}{r} = \frac{\dsp{N - \frac{1}{2}}}{Nt} + \frac{1}{2q},
 $$
 hence $\dsp{r = \frac{8n(n-1)q}{7n^2 - 2n - (6n+4)q}}$, which is a better value than $r = 2$ if $\dsp{q \geqslant \frac{7 n^2 - 2n}{4 n^2 + 2n + 4}}$.
\end{itemize}
\end{proof}

\section{The vector equation}\label{secVect}
In this section, we study Equation \eqref{eqVector}. Our aim is to prove the
following result:

\begin{proposition}\label{propVector}
Assume given a Riemannian metric $g\in W^{2,p}(M,S_2 M)$, $\gamma \in 
L^q(M,T^*M)$ and $\omega \in W^{1-\frac{1}{q}, q}(\partial M, T^*M)$. If $g$ has no non trivial conformal Killing vector field,
there exists a unique $W \in W^{2, q} (M,TM)$ solution to the vector equation with boundary condition:
\begin{equation}\label{eqVectorBoundary}
  \begin{cases} 
\displaystyle{\DeltaL W = \gamma ,}\\[3mm]
\displaystyle{(\bL W)(\nu, \cdot)=\omega.}
  \end{cases}
\end{equation}
 Furthermore, the mapping $(\gamma, \omega) \mapsto W$ is continuous.
\end{proposition}

Similar results have been obtained in \cite[Theorem 4.5]{HolstMeierTsogtgerel} 
with different regularity assumptions, in \cite[Proposition 6 and 
Theorem 3]{Maxwell} for asymptotically Euclidean manifolds and in 
\cite[Theorem 6.8]{Gicquaud} for asymptotically hyperbolic manifolds.
Note that, since we immediately rule out the existence of conformal Killing 
vector fields, our arguments will be slightly simpler.

Part of the proof of this proposition will be based on the Lax-Milgram theorem.
The main difficulty consists in proving that the quadratic form
$\displaystyle
W \mapsto \int_M  |\bL W|^2 \dd\mu^g
$
is coercive. This cannot be proven by the standard Bochner formula for
the vector Laplacian (see \cite[Chapter 2, §6]{Yano}):
\[
 \frac{1}{2} \int_M |\bL W|^2 \dd\mu^g = \int_M \left[|\nabla W|^2 + 
\frac{n-2}{n} (\divg W)^2 - \ric(W, W) \right] \dd\mu^g,
\]
as this formula is valid only for $W$'s whose support is disjoint from the 
boundary. Note that the full formula with the boundary terms is given by
\begin{align*}
 \frac{1}{2} \int_M |\bL W|^2 \dd\mu^g
 &= \int_M \left[|\nabla W|^2 + \frac{n-2}{n} (\divg W)^2 - \ric(W, W) \right] \dd\mu^g\\
 &\qquad + \int_{\partial M} \left(\bL W(W, \nu) - \<W, \nabla_\nu W\>
  - \frac{n-2}{2} \divg(W) \<W, \nu\> \right) \dd\mu^g
\end{align*}
so the boundary term does not agree with the boundary condition we are 
imposing. Instead, we rely on the method described in \cite[Section 
4]{HolstMeierTsogtgerel} which is based on \cite{DainKorn}. We reproduce here 
\cite[Corollary 1.2]{DainKorn} as we provide a complete proof with less 
regularity required ($g \in W^{1, p'}(M, S_2M)$ instead of $g \in C^1(M, S_2 M)$
for \cite{DainKorn}). The exponent $p'$ which appears in what follows will be 
taken later to be $\dsp{p' = \frac{np}{n-p}}$ if $\dsp{\frac{n}{2}<p<n}$ (or
$p'>n$ arbitrary if $p>n$).

\begin{lemma}\label{lmCoercivity}
Let $(M,g)$ be a compact Riemannian manifold with boundary such that $g\in 
W^{1,p'}(M, S_2 M)$ for some $p' > n$. There exist two positive constants $c_1$ and 
$c_2$ such that for any $W\in W^{1, 2}(M, TM)$
\[
 c_1 \|W\|^2_{W^{1, 2} (M, TM)} - c_2 \|W\|_{L^2 (M, TM)}^2 \leq \|\bL W\|_{L^2 (M, \Sring_2 M)}^2.
\]
\end{lemma}

\begin{proof}
We first choose $(U_i, \Phi_i)_{i\in I}$ a finite atlas on $M$ such that
$\Phi_i(U_i)$ is either $B_\varepsilon(0) \subset \bR^n$ or
$B_\varepsilon^+(0) = B_\varepsilon(0) \cap \{x_n \geq 0\}$ for some small
$\varepsilon > 0$ and such that:
\begin{itemize}
 \item For all $i \in I$, $\lambda^{-1} \delta \leq (\Phi_i)_* g \leq \lambda
\delta$,  where $\delta$ denotes the Euclidean metric. The value of $\lambda$ will be 
 chosen later.
 \item There exists a constant $\Lambda > 0$ such that $\|\partial ((\Phi_i)_* 
g)\|_{L^{p'}(\Phi_i(U_i))} \leq \Lambda$.
\end{itemize}
Such coordinates exist in a neighborhood of any point of $M$ since $g$ 
is in $W^{1, p'}(M, S_2 M)$ so, in particular, $g$ is H\"older continuous.
We also choose a family of functions $(\chi_i)_{i\in I}\in C^\infty (M, \Rb)$ such 
that $\displaystyle{\sum_{i\in I}\chi_i^2 =1}$ and $\forall i\in I, \supp 
(\chi_i) \subset U_i$. We thus get, for any $W \in W^{1, 2}(M, TM)$,

\begin{align}
\int_M |\bL W|^2 \dd\mu^g
 &= \int_M \left|\sum_{i\in I}\bL(\chi_i^2 W)\right|^2 \dd\mu^g\label{eqPartition}\\
 &=\sum_{i\in I} \int_M \left|\bL(\chi_i^2 W)\right|^2 \dd\mu^g + 
2 \sum_{i < j}\int_M \<\bL(\chi_i^2 W),\bL( \chi_j^2 W)\>\dd\mu^g.\nonumber
\end{align}

Straightforward calculations based on Leibniz' rule and 
symmetry and trace-freeness of $\bL V$ for any vector field $V$, give us
\begin{align*}
 \<\bL(\chi_i^2 W),\bL( \chi_j^2 W)\>
 & = \chi_i \chi_j \Big[\<\bL(\chi_i W),\bL( \chi_j W)\>\\
 &\qquad + 2 \<\bL(\chi_i W), \dd\chi_j \otimes W\>
  + 2 \<\bL(\chi_j W), \dd\chi_i \otimes W\>\\
 &\qquad + 2 \<\dd\chi_, \dd\chi_j\> |W|^2
   + 2\frac{n-2}{n} \<\dd\chi_i, W\>\<\dd\chi_j, W\>\Big],\\
 |\bL (\chi_i \chi_j W)|^2
 & = \chi_i \chi_j \Big[\<\bL(\chi_i W),\bL( \chi_j W)\>\\
 &\qquad + 2 \<\dd\chi_i , \dd\chi_j\> |W|^2
   + 2\frac{n-2}{n} \<\dd\chi_i, W\>\<\dd\chi_j, W\>\Big]\\
 &\qquad + 2 \chi_i^2 \<\bL (\chi_j W), \dd\chi_j\otimes W\>
   + 2\chi_j^2 \<\bL (\chi_i W), \dd\chi_i\otimes W\>.
\end{align*}
Subtracting both equalities, we get
\begin{align*}
 &\<\bL(\chi_i^2 W),\bL( \chi_j^2 W)\> - |\bL (\chi_i \chi_j W)|^2\\
 &\qquad  = 2 \chi_i \chi_j \<\bL(\chi_i W), \dd\chi_j \otimes W\>
  + 2 \chi_i \chi_j \<\bL(\chi_j W), \dd\chi_i \otimes W\>\\
 &\qquad \qquad - 2 \chi_i^2 \<\bL (\chi_j W), \dd\chi_j\otimes W\>
   - 2\chi_j^2 \<\bL (\chi_i W), \dd\chi_i\otimes W\>\\
 &\qquad = 2 \< \chi_j \bL(\chi_i W) - \chi_i \bL(\chi_j W),
 (\chi_i  \dd\chi_j - \chi_j \otimes \dd\chi_i) \otimes W\>\\
 &\qquad = 2 |W|^2  |\chi_i \dd\chi_j - \chi_j \dd\chi_i|^2
  - \frac{n+2}{n} \<W, \chi_i \dd\chi_j - \chi_j \dd\chi_i\>^2.
\end{align*}
As a consequence, we can replace $\<\bL(\chi_i^2 W),\bL( \chi_j^2 W)\>$
in \eqref{eqPartition} by $|\bL (\chi_i \chi_j W)|^2$ at the cost of 
introducing terms depending quadratically in $W$ but not on its covariant 
derivative. In particular, there exists a constant $c_0$ depending only on
$n$ and on all $\|\dd\chi_i\|_{L^\infty(M, \Rb)}$ such that, for any $W \in W^{1, 
2}(M, TM)$, we have
\begin{equation}\label{eqPartition2}
\begin{aligned}
&\int_M |\bL W|^2 \dd\mu^g + c_0 \|W\|_{L^2(M)}^2\\
&\qquad \geq \sum_{i\in I} \int_M \left|\bL(\chi_i^2 W)\right|^2 \dd\mu^g + 
2 \sum_{i < j} \int_M \left|\bL(\chi_i \chi_j W)\right|^2 \dd\mu^g.
\end{aligned}
\end{equation}
The main result of \cite{DainKorn} is that, for any bounded Lipschitz domain 
$\Omega$ in $\bR^n$, there exists a constant $c > 0$ such that, for any vector 
field $V \in W^{1, 2}(\Omega, T \Rb^n)$,
\begin{equation}\label{eqKornRn}
 \|V\|_{W^{1, 2}(\Omega, T \Rb^n)} \leq c \left(\int_{\Omega} |\bL_{\delta} V|_\delta^2 
\dd\mu^\delta + \int_{\Omega} |V|_\delta^2 \dd\mu^\delta\right).
\end{equation}
We will choose $\Omega$ to be either $B_\varepsilon(0)$ or $B_\varepsilon^+(0)$
in what follows and use $\Omega$ to denote one or another of these subsets.

Our next task is to prove that the inequality \eqref{eqKornRn} still holds with $\delta$
replaced by $(\Phi_i)_* g$ (that we will denote by $g$ for simplicity) and $V$ 
with compact support in either $B_\varepsilon(0)$ or $B_\varepsilon^+(0)$. We point 
out that, as we will choose $V$ to be either $\chi_i^2 W$ or $\chi_i \chi_j W$, it has compact support and can 
be extended by zero to larger balls (or half balls), in particular to 
$B_1(0)$ or $B_1^+(0)$ (assuming $\varepsilon \leqslant 1$). This shows that the
constant $c$ in \eqref{eqKornRn} can be chosen independently of $\varepsilon \leqslant 1$.

Note that, as we will use both metrics $g$ and $\delta$, we will need to keep 
track of which metric is used. As a consequence, in the rest of the proof, the 
metric with which norms, conformal Lie derivatives and so on will be indicated
in subscript. The only exception will be the notation $\partial$ for the 
derivative with respect to $\delta$.

The conformal Killing operators $\bL_g$ and $\bL_\delta$ are related by
\begin{equation}\label{eqDiffKilling0}
\begin{aligned}
 (\bL_g V)_{ij} &= (\bL_\delta V)_{ij}
+ (g_{kj} \Gamma^k_{il} + g_{ik} \Gamma^k_{jl} - \frac{2}{n} 
\Gamma^k_{kl} g_{ij}) V^l\\
 &\qquad + (g_{kj} - \delta_{kj}) \partial_i V^k + (g_{ik} - \delta_{ik}) 
\partial_j V^k - \frac{2}{n} (g_{ij} - \delta_{ij}) \partial_k V^k,
\end{aligned}
\end{equation}
where $\dsp{\Gamma_{ij}^{k}=\frac{1}{2}g^{kl}(\partial_i g_{lj}+\partial_j g_{il}-\partial_l g_{ij})}$ denotes the Christoffel symbols of $g$. The first two terms on the second line come from the fact that $V$ is a vector field so we need to take into consideration which metric is used.

Taking norms, we conclude that, for some constant $C > 0$, we have
\[
 |\bL_g V|^2_\delta \geq \frac{1}{2} |\bL_g V|^2_\delta - C 
\left(|g-\delta|^2_\delta |\partial V|^2_\delta + |\partial g|_\delta^2 
|V|^2_\delta\right).
\]
Integrating this inequality over $\Omega$, we get, for some constant $\kappa = 
\kappa(\lambda)$ that may vary from line to line but which is such that 
$\kappa(\lambda)$ tends to $1$ when $\lambda$ tends to $1$.
\begin{align}
 &\int_{\Omega} |\bL_g V|_g^2 \dd\mu^g \label{eqDiffKilling}\\
 &\qquad\geq \kappa \int_{\Omega} |\bL_g V|_\delta^2 \dd\mu^\delta\nonumber\\
 &\qquad \geq \frac{\kappa}{2} \int_{\Omega} |\bL_\delta V|^2_\delta 
\dd\mu^\delta  - \kappa C 
\left[\left(\sup_\Omega |g-\delta|^2_\delta\right) \int_{\Omega} |\partial 
V|_\delta^2 \dd\mu^\delta + \int_{\Omega} |\partial g|_\delta^2 
|V|^2_\delta \dd\mu^\delta\right]\nonumber\\
 &\qquad\geq \frac{\kappa}{2} \int_{\Omega} |\bL_\delta V|^2_\delta 
\dd\mu^\delta\nonumber\\
 &\qquad\qquad - \kappa C 
\left[(\lambda-1)^2 \int_{\Omega} |\partial V|_\delta^2 \dd\mu^\delta +
\left(\int_{\Omega} |\partial g|_\delta^{p'} \dd\mu^\delta\right)^{2/p'} 
\left(\int_{\Omega} |V|^q_\delta \dd\mu^\delta\right)^{2/q}\right],\nonumber
\end{align}
where $q$ is such that $\displaystyle \frac{2}{p'} + \frac{2}{q} = 1$. As $p' > 
n$, we have that $q < N$. Let $\theta \in (0, 1)$ be such that
$\displaystyle \frac{1}{q} = \frac{1-\theta}{2} + \frac{\theta}{N}$. By Young's inequality, we get that, for any $\mu > 0$,
\[
 \left(\int_{\Omega} |V|^q_\delta \dd\mu^\delta\right)^{2/q}
 \leq \theta \mu \left(\int_{\Omega} |V|^N_\delta \dd\mu^\delta\right)^{2/N}
  + (1-\theta) \mu^{-\theta/(1-\theta)} \int_{\Omega} |V|^2_\delta 
\dd\mu^\delta.
\]
From the Sobolev inequality, there exists a constant $s > 0$ such that
\[
 \left(\int_{\Omega} |V|^N_\delta \dd\mu^\delta\right)^{2/N}
 \leq s \|V\|_{W^{1, 2}(\Omega, T \Rb^n)}.
\]
Note that, as for the constant in Korn's inequality, the constant $s$ is 
independent of $\varepsilon$. As a consequence, the previous inequality can be 
transformed into
\[
 \left(\int_{\Omega} |V|^q_\delta \dd\mu^\delta\right)^{2/q}
 \leq \theta \mu s \|V\|_{W^{1, 2}(\Omega, T \Rb^n)}^2
  + (1-\theta) \mu^{-\theta/(1-\theta)} \int_{\Omega} |V|^2_\delta 
\dd\mu^\delta.
\]
Finally, the estimate \eqref{eqDiffKilling} becomes
\begin{align*}
 &\int_{\Omega} |\bL_g V|_g^2 \dd\mu^g
 - \frac{\kappa}{2} \int_{\Omega} |\bL_\delta V|^2_\delta \dd\mu^\delta\\
 &\qquad\geq - \kappa C 
\left[(\lambda-1)^2 \int_{\Omega} |\partial V|_\delta^2 \dd\mu^\delta +
\Lambda^2 
\left(\theta \mu s \|V\|_{W^{1, 2}(\Omega, T \Rb^n)}^2
  + (1-\theta) \mu^{\frac{-\theta}{1-\theta}} \int_{\Omega} |V|^2_\delta 
\dd\mu^\delta\right)\right].
\end{align*}
From Korn's identity \eqref{eqKornRn}, we conclude
that
\begin{align*}
 &\int_{\Omega} |\bL_g V|_g^2 \dd\mu^g\\
 &\qquad\geq \kappa \left(\frac{c}{2} - C(\lambda-1)^2 - C 
\Lambda^2 \theta \mu s\right)\|V\|_{W^{1, 2}(\Omega)}^2
- C' \int_{\Omega} |V|^2_\delta \dd\mu^\delta.
\end{align*}
for some large constant $C'$. So we see that, choosing $\lambda$ close 
enough to $1$ and $\mu > 0$ sufficiently small, we get
\begin{equation}\label{eqKorng1}
 \int_{\Omega} |\bL_g V|_g^2 \dd\mu^g
 \geq \frac{c \kappa}{4} \|V\|_{W^{1, 2}(\Omega, T \Rb^n)}^2
- C' \int_{\Omega} |V|^2_\delta \dd\mu^\delta.
\end{equation}
Combining \eqref{eqKorng1} with \eqref{eqPartition2}, we obtain that, for 
some new constant $c_0$ and some $\mu > 0$, we have
\begin{equation}\label{eqPartition3}
\begin{aligned}
&\int_M |\bL W|^2 \dd\mu^g + c_0 \|W\|_{L^2(M, T \Rb^n)}^2\\
&\qquad \geq \mu \left(\sum_{i\in I} \int_M \left|\nabla (\chi_i^2 W)\right|^2 
\dd\mu^g + 
2 \sum_{i < j} \int_M \left|\nabla (\chi_i \chi_j W)\right|^2 \dd\mu^g\right).
\end{aligned}
\end{equation}
Note that we have taken the freedom to replace $\|V\|_{W^{1, 2}(\Omega, T\Rb^n)}^2$ by
$\displaystyle \int_M \left|\nabla V\right|^2 \dd\mu^g$ (with $V = \chi_i^2 W$ or
$V = \chi_i \chi_j W$) as calculations similar to the ones after Equation 
\eqref{eqDiffKilling} allow us to conclude that, under our assumptions
\[
 \int_\Omega \left|\nabla V\right|^2 \dd\mu^g \lesssim \|V\|_{W^{1, 2}(\Omega, T\Rb^n)}
\]
so \eqref{eqPartition3} remains valid upon redefinnig $\mu$.

Our last task is to rewind the construction we performed to get 
\eqref{eqPartition2} with $\bL_g$ replaced by $\nabla$. Calculations are 
similar 
to those leading to \eqref{eqPartition2} so we only indicate the result:
\begin{equation}\label{eqPartition4}
\begin{aligned}
 \< \nabla(\chi_i^2 W), \nabla(\chi_j^2 W)\>
 &= |\nabla(\chi_i \chi_j W)|^2 - |\chi_i \dd \chi_j - \chi_j \dd \chi_i|^2 |W|^2\\
 &\leq |\nabla(\chi_i \chi_j W)|^2
\end{aligned}
\end{equation}
Thus, from Equation \eqref{eqPartition3}, it follows that
\[
\begin{aligned}
\int_M |\nabla W|^2 \dd\mu^g
 &= \int_M \left|\sum_{i\in I}\nabla (\chi_i^2 W)\right|^2 \dd\mu^g\\
 &=\sum_{i\in I} \int_M \left|\nabla (\chi_i^2 W)\right|^2 \dd\mu^g + 
2 \sum_{i < j}\int_M \<\nabla (\chi_i^2 W), \nabla( \chi_j^2 W)\>\dd\mu^g\\
 &\leq \sum_{i\in I} \int_M \left|\nabla (\chi_i^2 W)\right|^2 \dd\mu^g + 
2 \sum_{i < j}\int_M \left|\nabla (\chi_i \chi_j W)\right|^2\dd\mu^g\\
 &\leq \frac{1}{\mu} \left(\int_M |\bL W|^2 \dd\mu^g + c_0 
\|W\|_{L^2(M)}^2\right).
\end{aligned}
\]
\end{proof}

We are now in a position to prove Proposition \ref{propVector}. To this end,
we introduce the functional $J_0$ defined on $W^{1, 2}(M, TM)$ as follows:
\begin{equation}\label{eqDefJ}
 J_0(V) = \int_M |\bL V|^2 \dd\mu^g.
\end{equation}
As a consequence of Lemma \ref{lmCoercivity}, we have the following result:

\begin{lemma}\label{lmCoercivity2}
Under the assumptions of Proposition \ref{propVector}, the functional $J_0$
is coercive.
\end{lemma}

\begin{proof}
 We assume by contradiction that for all integer $k > 0$, there exists
 $V_k \in W^{1, 2}(M, TM)$ such that
 \[
  J_0(V_k) \leq \frac{1}{k} \|V_k\|_{W^{1, 2}(M, TM)}^2.
 \]
 Upon rescaling $V_k$, we can assume further that $\|V_k\|_{L^2(M, TM)} = 1$.
 From Lemma \ref{lmCoercivity}, we have
 \[
  c_1 \|V_k\|_{W^{1, 2}(M, TM)}^2 - c_2 \|V_k\|_{L^2(M, TM)}^2
  \leq J_0(V_k) \leq \frac{1}{k} \|V_k\|_{W^{1, 2}(M, TM)}^2.
 \]
 In particular, if $\dsp{k \geq \frac{2}{c_1}}$, we get
 \[
  \frac{c_1}{2} \|V_k\|_{W^{1, 2}(M, TM)}^2 \leq c_2 \|V_k\|_{L^2(M, TM)}^2
   = c_2
 \]
 showing that $(V_k)_k$ is bounded in $W^{1, 2}(M, TM)$. By compactness of the
 embedding $W^{1, 2}(M, TM) \hookrightarrow L^2(M, TM)$, we can further assume
 that the sequence $(V_k)$ converges to some $V_\infty \in L^2(M, TM)$. In
 particular,
 \[
  \|V_\infty\|_{L^2(M, TM)} = \lim_{k \to \infty} \|V_k\|_{L^2(M, TM)} = 1.
 \]
 So $V_\infty \not\equiv 0$. Using once again Lemma \ref{lmCoercivity}, we get that,
 for all $\dsp{k, k' \geq \frac{2}{c_1}}$,
 \begin{align*}
  &c_1 \|V_k - V_{k'}\|_{W^{1, 2}(M, TM)}^2 - c_2 \|V_k - V_{k'}\|_{L^2(M, 
TM)}^2\\
  &\qquad\leq J_0(V_k - V_{k'})\\
  &\qquad\leq 2 (\|\bL V_k\|_{L^2(M, TM)}^2 + \|\bL V_{k'}\|_{L^2(M, TM)}^2)\\
  &\qquad\leq 2 \left(\frac{1}{k} \|V_k\|_{W^{1, 2}(M, TM)}^2 + \frac{1}{k'} 
\|V_{k'}\|_{W^{1, 2}(M, TM)}^2\right)\\
  &\qquad \leq\frac{4 c_2}{c_1} \left(\frac{1}{k} + \frac{1}{k'}\right).
 \end{align*}
 As a consequence, choosing $k$ and $k'$ large enough, we can make the norm
 $\|V_k - V_{k'}\|_{W^{1, 2}(M, TM)}$ as small as we want. The sequence 
$(V_k)_k$ is thus Cauchy in $W^{1, 2}(M, TM)$ and converges to $V_\infty \in 
W^{1, 2}(M, M)$. We have
 \[
  J_0(V_\infty)
  = \lim_{k \to \infty} J_0(V_k)
  \leq \lim_{k \to \infty} \frac{2 c_2}{k c_1 } = 0.
 \]
 So $V_\infty$ is a non-trivial solution to $\bL V = 0$. This contradicts
 the assumptions of Proposition \ref{propVector}. It follows that $J_0$ is 
coercive.
\end{proof}

We can finally prove Proposition \ref{propVector} in a succession of claims.
For any $q \in (1, p]$, we introduce the operator
\[
 \cP_q: W^{2, q}(M, TM) \to L^q(M, T^* M) \times W^{1-\frac{1}{q}, q}(\partial M, T^*M)
\]
defined by
\[
 \cP_q(W) \definedas (\DeltaL W,~\bL W(\nu, \cdot)).
\]

\begin{claim}\label{clRegEllip}
 There exists a constant $C$ such that, for any $W \in W^{2, q}(M, TM)$, we 
have
 \[
  \left\|W\right\|_{W^{2, q}(M, TM)} \leq C 
\left(\left\|\DeltaL W\right\|_{L^q(M, T^*M)}
+ \left\|\bL W(\nu, \cdot)\right\|_{W^{1-\frac{1}{q}, q}(\partial M, T^*M)}
+ \left\|W\right\|_{L^q(M, TM)} \right).
 \]
\end{claim}

This claim is proven in a more general context in \cite[Proposition 5]{Maxwell}, see also 
\cite[Section 6]{Gicquaud} and is an application of the general theory developed in
\cite{ADN} (see also \cite[Theorem 6.3.7]{Morrey}). The next claim is a 
particular case of \cite[Theorem 6.4.8]{Morrey}:

\begin{claim}\label{clImprovedRegularity}
 Given $q \in [2, p]$, assume that $W \in W^{1, 2}(M, TM)$ is a
 weak solution to the problem \eqref{eqVectorBoundary} with
 $\gamma \in L^q(M, T^*M)$ and $\omega \in W^{1-\frac{1}{q}, q}(\partial M, T^*M)$.
 Namely, $W$ satisfies that, for all $V \in W^{1, 2}(M, TM)$,
\begin{equation}\label{eqWeakVector}
 \frac{1}{2} \int_M \< \bL V, \bL W\> \dd\mu^g
  = \int_{\partial M} \<V, \omega\> \dd\mu^g - \int_M \<V, \gamma\> \dd\mu^g.
\end{equation}
Then $W \in W^{2, q}(M, TM)$.
\end{claim}

Remark that, when $q=2$, there is an elegant way to prove this result using difference quotients, see \cite[Section IX.6]{Brezis}

Claim \ref{clImprovedRegularity} implies, in particular, that any
$W \in \ker(\cP_q)$ belongs to $W^{2, p}(M, TM)$. Thus, for such a $W$,
taking $V = W$ in Formula \eqref{eqWeakVector} yields
\[
 \int_M |\bL W|^2 \dd\mu^g = 0,
\]
so $W$ is a conformal Killing vector. As we excluded the possibility for $M$ to 
have non-zero conformal Killing vector fields, we obtain the next claim:

\begin{claim}\label{clInjectivity}
 The operator $\cP_q$ is injective for all $q \in (0, p]$.
\end{claim}

We now turn our attention to proving that $\cP_q$ is surjective. This is the 
content of the following two claims:

\begin{claim}\label{clClosedImage}
For any $q = \in (1, p]$, there exists a constant $C'$ such that, for any $W 
\in W^{2, q}(M, TM)$, we have
 \begin{equation}\label{eqRegEllip}
  \left\|W\right\|_{W^{2, q}(M, TM)} \leq C'
\left(\left\|\DeltaL W\right\|_{L^q(M, T^*M)}
+ \left\|\bL W(\nu, \cdot)\right\|_{W^{1-\frac{1}{q}, q}(\partial M, T^*M)}\right).
 \end{equation}
 In particular, $\cP_q$ has closed range.
\end{claim}

\begin{proof}
 The argument goes by contradiction. Assume that there exists no such constant 
$C'$. Then, for any integer $k \geq 1$, there exists a vector field
$W_k \in W^{2, q}(M, TM)$ such that
 \begin{equation}\label{eqContradiction}
  \left\|W_k\right\|_{W^{2, q}(M, TM)} \geq k
\left(\left\|\DeltaL W_k\right\|_{L^q(M, T^*M)}
+ \left\|\bL W_k(\nu, \cdot)\right\|_{W^{1-\frac{1}{q}, q}(\partial M, T^*M)}\right).
 \end{equation}
 Without loss of generality, we can assume that $\left\|W_k\right\|_{L^q(M, 
TM)} = 1$. From Claim \ref{clRegEllip}, we conclude that
 \[
  \left\|W_k\right\|_{W^{2, q}(M, TM)}
  \leq \frac{C}{k} \left\|W_k\right\|_{W^{2, q}(M, TM)}
+ C \left\|W_k\right\|_{L^q(M, TM)}.
 \]
 In particular, if $k \geq 2C$, we have
 \[
  \left\|W_k\right\|_{W^{2, q}(M, TM)}
  \leq 2 C \left\|W_k\right\|_{L^q(M, TM)} = 2 C.
 \]
 So the sequence $(W_k)_k$ is bounded in $W^{2, q}(M, TM)$. As the embedding
 $W^{2, q}(M, TM) \hookrightarrow L^q(M, TM)$ is compact, we can assume that 
the sequence $(W_k)_k$ converges in $L^q(M, TM)$ to some $W_\infty \in L^q(M, 
TM)$.
 From Claim \ref{clRegEllip}, we have for any pair of integers $(k, l)$ with
 $k, l \geq 2C$,
 \begin{align*}
  \left\|W_k - W_l\right\|_{W^{2, q}(M, TM)}
  & \leq C \left(\left\|\DeltaL W_k - \DeltaL W_l\right\|_{L^q(M, 
T^*M)}\right.\\
  &\qquad \left. + \left\|\bL W_k(\nu, \cdot) - \bL W_l(\nu, 
\cdot)\right\|_{W^{1-\frac{1}{q}, q}(\partial M, T^*M)}
  + \left\|W_k - W_l\right\|_{L^q(M, TM)} \right)\\
  & \leq C \left(\left\|\DeltaL W_k\right\|_{L^q(M, 
T^*M)} + \left\|\DeltaL W_l\right\|_{L^q(M, T^*M)}\right.\\
  &\qquad + \left\|\bL W_k(\nu, \cdot)\right\|_{W^{1-\frac{1}{q}, q}(\partial M, 
T^*M)} + \left\|\bL W_l(\nu, \cdot)\right\|_{W^{1-\frac{1}{q}, q}(\partial M, T^*M)}\\
  & \qquad + \left.\left\|W_k - W_l\right\|_{L^q(M, TM)} \right)\\
  & \leq C \left(\frac{1}{k} \left\|W_k\right\|_{W^{2, q}(M, TM)}
   + \frac{1}{l} \left\|W_l\right\|_{W^{2, q}(M, TM)}\right.\\
  &\qquad + \left. \left\|W_k - W_l\right\|_{L^q(M, TM)}\right)\\
  &\leq 2 C^2 \left(\frac{1}{k} + \frac{1}{l}\right)
   + C \left\|W_k - W_l\right\|_{L^q(M, TM)}.
 \end{align*}
 So the sequence $(W_k)$ is Cauchy in $W^{2, q}(M, TM)$. Thus, we have
 $W_\infty \in W^{2, q}(M, TM)$ and, passing to the limit $k \to \infty$ in
 Equation \eqref{eqContradiction}, we conclude that
 $W_\infty$ satisfies $\DeltaL W_\infty = 0$ and
 $\bL W_\infty(\nu, \cdot) = 0$. From the discussion preceeding Claim 
\ref{clInjectivity}, we know that $W_\infty \equiv 0$.
 However, as $\|W_k\|_{L^q(M, TM)} = 1$ for all $k$, we have that 
$\|W_\infty\|_{L^q(M, TM)} = 1$. This gives the desired contradiction and 
proves the inequality \eqref{eqRegEllip}.

Closure of the range is then a simple consequence of Inequality \eqref{eqRegEllip}.
Assume  given a converging sequence of elements $(\gamma_k, \omega_k)$ in the 
range of $\cP_q$ that converges to some
$(\gamma_\infty, \omega_\infty) \in L^q(M, T^*M) \times W^{1-1/q, q}(\partial 
M, T^*M)$. As $\cP_q$ is injective, there exists a unique $W_k \in W^{2, q}(M, 
TM)$ such that $\cP_q(W_k) = (\gamma_k, \omega_k)$. The inequality 
\eqref{eqRegEllip} implies that, for any pair of integers $(k, l)$, we have
\begin{align*}
 &\left\|W_k - W_l\right\|_{W^{2, q}(M, TM)}\\
 &\qquad \leq C' \left(\left\|\DeltaL (W_k-W_l)\right\|_{L^q(M, T^*M)}
  + \left\|\bL (W_k - W_l)(\nu, \cdot)\right\|_{W^{1-\frac{1}{q}, q}(\partial M, 
T^*M)}\right)\\
 & \qquad \leq C' \left(\left\|\gamma_k - \gamma_l\right\|_{L^q(M, T^*M)}
  + \left\|\omega_k - \omega_l\right\|_{W^{1-\frac{1}{q}, q}(\partial M, 
T^*M)}\right).
\end{align*}
As the sequence $((\gamma_k, \omega_k))_k$ converges, it is Cauchy so the right-
hand side of the previous inequality tends to zero as $k$ and $l$ tend to 
infinity. This shows that the sequence $(W_k)_k$ is also Cauchy. Since $W^{2, 
q}(M, TM)$ is a Banach space, the sequence $(W_k)_k$ tends to some  $W_\infty 
\in W^{2, q}(M, TM)$. As $\cP_q$ is continuous, we have $\cP_q(W_\infty) = 
(\gamma_\infty, \omega_\infty)$ showing that $(\gamma_\infty, \omega_\infty)$ 
is in the range of $\cP_q$. This concludes the proof of the fact that the range 
of $\cP_q$ is closed.
\end{proof}

\begin{claim}\label{clDenseImage}
For any $q = \in (1, p]$, $\cP_q$ has dense range.
\end{claim}

\begin{proof}
This is now that we use Lemma \ref{lmCoercivity2}. Assume first that
$q \geq 2$. We can use the Lax-Milgram theorem to conclude that for any
$(\gamma, \omega) \in L^q(M, T^M) \times W^{1-\frac{1}{q}, q}(\partial M, T^*M)$,
there is a weak solution to the problem \eqref{eqVectorBoundary}, that is to 
say a $W \in W^{1, 2}(M)$ solving \eqref{eqWeakVector} for all
$V \in W^{1, 2}(M, TM)$. From Claim \ref{clImprovedRegularity}, we conclude 
that $W \in W^{2, q}(M, TM)$ meaning that $(\gamma, \omega)$ is in the range of 
$\cP_q$.

In the case $q < 2$, the strategy is to remark that $L^2(M, T^*M) \times 
W^{\frac{1}{2}, 2}(M, T^*M)$ is dense in $L^q(M, T^*M) \times W^{1-\frac{1}{q}, q}(M, T^*M)$.
As we know that $\cP_2$ is surjective, we conclude that the range of $\cP_q$ 
contains a dense subspace. In particular, the range of $\cP_q$ is dense.
\end{proof}

It then follows from claims \ref{clClosedImage} and 
\ref{clDenseImage} that the range of $\cP_q$ is the whole of
$L^q(M, T^*M) \times W^{1-\frac{1}{q}, q}(\partial M, T^* M)$. As $\cP_q$ is also 
injective (Claim \ref{clInjectivity}), we conclude from the open mapping theorem 
that $\cP_q$ is a linear isomorphism between $W^{2, q}(M, TM)$ and
$L^q(M, T^*M) \times W^{1-\frac{1}{q}, q}(\partial M, T^* M)$. This proves Proposition 
\ref{propVector}.

As an application of Proposition \ref{propVector}, we prove an analog of York's
decomposition in the context of compact manifolds with boundary:

\begin{proposition}\label{propYork}
For any symmetric traceless 2-tensor $T \in W^{1, p}(M, \Sring_2 M)$, there
exists a unique pair
\[
 (\sigma, W) \in W^{1, p}(M, \Sring_2 M) \times W^{2, p}(M, TM)
\]
with $\sigma$ being divergence-free (i.e. a TT-tensor) satisfying
\[
 \sigma(\nu, \cdot) \equiv 0\text{ on } \partial M.
\]
such that $T = \sigma + \bL W$. This decomposition is $L^2$-orthogonal:
\[
 \int_M \< \sigma, \bL W\> \dd\mu^g = 0.
\]
In particular, $\displaystyle
\int_M |\sigma + \bL W|^2 \dd\mu^g = \int_M \left(|\sigma|^2 + |\bL W|^2\right) \dd\mu^g .
$
\end{proposition}

\begin{proof}
This is a simple application of what we have proven so far. We rewrite the
equation $\divg_g(\sigma) = 0$ as
\begin{equation}\label{eqYorkInner}
 \DeltaL W = \divg_g(\bL W) = \divg_g(T - \sigma) = \divg(T)
\end{equation}
and the boundary condition for $\sigma$ as
\begin{equation}\label{eqYorkBoundary}
 \bL W(\nu, \cdot) = (T - \sigma)(\nu, \cdot) = T(\nu, \cdot).
\end{equation}
From the regularity of $T$, we have $\divg_g(T) \in L^p(M, T^*M)$ and
$T(\nu, \cdot) \in W^{1-\frac{1}{p}, p}(\partial M, T^*M)$ so Proposition \ref{propVector}
applies showing that there exists a unique $W$ satisfying both \eqref{eqYorkInner}
and \eqref{eqYorkBoundary}. The orthogonality of $\sigma$ and $\bL W$ follows from
a simple calculation:
\begin{align*}
 \int_M \< \sigma, \bL W\> \dd\mu^g
 &= 2 \int_M \< \sigma, \nabla W\> \dd\mu^g\\
 &= 2 \int_M \left(\divg_g(\sigma(W, \cdot)) - \<\divg_g(\sigma), W\>\right) \dd\mu^g\\
 &= 2 \int_{\partial M} \sigma(W, \nu) \dd\mu^g - \int_M \<\divg_g(\sigma), W\> \dd\mu^g\\
 &= 0.
\end{align*}
The last line follows because of the conditions imposed on $\sigma$.
\end{proof}

\section{The coupled system}\label{secCoupled}
We now study the full system and prove the existence of solutions by means of
the Schauder fixed point theorem which we recall now in a slightly different
form:

\begin{theorem}\label{thmSchauder}
Let $X$ be a Banach space and $\Phi: X \to X$ a continuous mapping. Let 
$\Omega$ be a closed convex set such that $\Phi(\Omega)$ is relatively compact (in $X$)
and contained in $\Omega$. Then $\Phi$ admits a fixed point in $\Omega$.
\end{theorem}

In the classical statement of the theorem (see e.g.
\cite[Theorem 11.1]{GilbargTrudinger}), $\Omega$ is assumed to be convex and
compact. Our statement is an easy consequence of the classical one as we can
replace the subset $\Omega$ by the closed convex hull $\Omegatil$ of
$\Phi(\Omega)$ which is a compact subset of $\Omega$ (see e.g. \cite[Theorem 3.24]{Rudin})
so, in particular, $\Phi(\Omegatil) \subset \Phi(\Omega)) \subset \Omegatil$.
This formulation is, however, more convenient for us.

The overall strategy follows that of \cite{GicquaudHabilitation}, see
\cite{cang1} for the original method. We choose for $X$ the space $W^{1, 2}(M, TM)$
of vector fields with Sobolev regularity. The mapping $\Phi$ is then constructed
as follows: Given $V \in W^{1, 2}(M, TM)$, we set $A = |\sigma + \bL V|\in L^2 (M,\Rb)$ and solve
for $\phi \in W^{1, 2}(M, \bR)$ the weak formulation of the Lichnerowicz equation
with apparent horizon boundary condition (see Proposition \ref{propLich}).
We then solve for $W$ the (weak form of the) vector equation \eqref{eqVector} with
boundary condition \eqref{eqCondVector} and set $\Phi(V) = W$.

We first prove that the mapping $\Phi$ satisfies the assumptions of the Schauder fixed
point theorem (Proposition \ref{propPhi}) and find some invariant subset $\Omega$ for
$\Phi$ (Proposition \ref{propInvariant}). As a consequence of Theorem \ref{thmSchauder},
we conclude that there exists a fixed point for $\Phi$, i.e. a weak solution to the
problem \ref{BCCE}. 

We place ourselves under the assumptions of Theorrem \ref{thmMain} and thus assume $\sigma\not\equiv 0$ or $\xi\not\equiv 0$.

\begin{proposition}\label{propPhi}
Assume that $\tau \in W^{1,t}(M, \bR) \cap L^\infty(M, \bR)$ for some
$\dsp{t > \frac{2n(n-1)}{3n-2}}$. Then $\Phi$ is well-defined, continuous and compact.
\end{proposition}

\begin{proof}
From our assumptions, the mapping $V \mapsto A \definedas |\sigma + \bL V|$ is continuous
as a map from $W^{1, 2}(M, TM)$ to $L^2(M, \bR)$. We now apply Proposition \ref{propContinuity}
and get that $A \mapsto \phi^N$ is continuous and compact as a map to
$L^a(M, \bR) \times L^b(\partial M, \bR)$, where $a$ and $b$ are in the ranges
given there. We set
\[
 \gamma \definedas \frac{n-1}{n} \phi^N \dd\tau \quad \text{and} \quad
 \omega \definedas f \phi^N \nu^\flat + \xi,
\]
where $\dsp{f \definedas \frac{n-1}{n} \tau - \frac{\Theta_-}{2} \in L^\infty(\partial M, \bR)}$.
We have that $\phi^N \mapsto \gamma$ is a continuous linear map from $L^a(M, \bR)$ to
$L^{\frac{2n}{n+2}}(M, T^*M)$ as it follows from an application of H\"older's inequality and our choice for $t$.

Similarly, we get that $\phi^N \mapsto f \phi^N \nu^\flat + \xi$ is a continuous linear
map to $L^{\frac{2(n-1)}{n}}(\partial M, \bR)$. Due to our choices of regularity for $\gamma$ and $\omega$, we get that the linear form
\[
 W \mapsto \int_{\partial M} \<W, \omega\> \dd\mu^g - \int_M \<W, \gamma\> \dd\mu^g
\]
is continuous over $W^{1, 2}(M, TM)$. It follows from the Lax-Milgram theorem that
there exists a unique $W \in W^{1, 2}(M, TM)$ such that, for all
$X \in W^{1, 2}(M, TM)$, we have
\begin{equation}\label{eqWeakVector2}
 \frac{1}{2} \int_M \< \bL X, \bL W\> \dd\mu^g
  = \int_{\partial M} \<X, \omega\> \dd\mu^g - \int_M \<X, \gamma\> \dd\mu^g,
\end{equation}
namely $W$ is a weak solution to the problem \eqref{eqVector}-\eqref{eqCondVector}
(see Equation \eqref{eqWeakVector} in Claim \ref{clImprovedRegularity}).
$W$ depends continuously on $(\gamma, \omega)$. All in all, we have shown that
$V \mapsto W \definedas \Phi(V)$ is well-defined and continuous. It is also compact
as $A \mapsto \phi^N$ is compact.
\end{proof}

\begin{proposition}\label{propInvariant}
There exists positive constants $\varepsilon$ and $\mu$ such that, if
\[
 \|\sigma\|_{L^2(M, S_2(M))} \leq \varepsilon \quad\text{and}\quad
 \|\xi\|_{L^{\frac{2(n-1)}{2}}(\partial M, T^*M)} \leq \mu,
\]
there exists $R > 0$ such that the set
\[
 \Omega_R \definedas \left\{W \in W^{1, 2}(M, TM), \int_M |\bL W|^2 \dd\mu^g \leq R^2\right\}
\]
is stable by $\Phi$ for some well chosen $R$: $\Phi(\Omega_R) \subset \Omega_R$.
\end{proposition}

\begin{proof}
Let $R > 0$ be arbitrary for the moment. We first note that, if
$W \in \Omega_R$ and $A \definedas |\sigma + \bL W|$, we have
\[
 \int_M A^2 \dd\mu^g = \int_M |\sigma + \bL W|^2 \dd\mu^g = \int_M (|\sigma|^2 + |\bL W|^2) \dd\mu^g \leq R^2 + \varepsilon^2.
\]
Let $\phi$ denote the solution to the Lichnerowicz equation \eqref{eqLichnerowicz}
with apparent horizon boundary condition \eqref{eqCondLich}.
From Proposition \ref{propImprovedRegL2}, together with the Sobolev and
trace inequalities, we have, for some positive constants $C$ and $c$,
\begin{align*}
 \|\phi^N\|_{L^{\fraccheloue}(M, \bR)}^{\frac{\fraccheloue}{N}}
 &= \left\|\phi^{\fraccheloue}\right\|_{L^N(M, \bR)} \leq C \|A\|_{L^2(M, \bR)} \leq C \sqrt{R^2 + \varepsilon^2},\\
 \|\phi^N\|_{L^q(\partial M, \bR)}^{\frac{\fraccheloue}{N}}
 &= \left\|\phi^{\fraccheloue}\right\|_{L^{\fraccheloue}(\partial M, \bR)} \leq c \|A\|_{L^2(M, \bR)} \leq c \sqrt{R^2 + \varepsilon^2},
\end{align*}
where we set $\dsp{q \definedas \frac{\dsp{\left(\fraccheloue\right)^2}}{N} = \frac{2 (n-1)^2}{n(n-2)}}$. So
\[
 \|\phi^N\|_{L^{\fraccheloue}(M, \bR)} \leq C^{\frac{N}{\fraccheloue}} (R^2 + \varepsilon^2)^{\frac{N}{N+2}}
 \quad\text{and}\quad
 \|\phi^N\|_{L^q(\partial M, \bR)} \leq c^{\frac{N}{\fraccheloue}} (R^2 + \varepsilon^2)^{\frac{N}{N+2}}.
\]
If we let $\dsp{\gamma = \frac{n-1}{n} \phi^N \dd\tau}$, we get
\begin{equation}\label{eqEstimateGamma}
 \|\gamma\|_{L^{\frac{2n}{n+2}}(M, \bR)} \leq \frac{n-1}{n} \|\phi^N\|_{L^{\fraccheloue}(M, \bR)} \|\dd\tau\|_{L^{\frac{2n(n-1)}{3n-2}}(M, \bR)} \leq C' (R^2 + \varepsilon^2)^{\frac{N}{N+2}},
\end{equation}
where $\dsp{C' = \frac{n-1}{n} \|\dd\tau\|_{L^{\frac{2n(n-1)}{3n-2}}(M, \bR)} C^{\frac{N}{\fraccheloue}}}$.

Similarly, if $\displaystyle \omega = \left(\frac{n-1}{n} \tau - \frac{\Theta_-}{2}\right) 
\phi^N \nu^\sharp + \xi$, we have
\begin{equation}\label{eqEstimateOmega}
\begin{aligned}
 \|\omega\|_{L^{\frac{2(n-1)}{n}}(\partial M, T^*M)} 
 & \leq \left\|\frac{n-1}{n} \tau - \frac{\Theta_-}{2}\right\|_{L^{\frac{2(n-1)^2}{n}}(\partial M, \bR)} \|\phi^N\|_{L^q(\partial M, \bR)}\\
 &\qquad\qquad + \|\xi\|_{L^{\frac{2(n-1)}{n}}(\partial M, T^*M)}\\
 & \leq c' (R^2 + \varepsilon^2)^{\frac{N}{N+2}} + \mu,
\end{aligned}
\end{equation}
with $\dsp{c' = \left\|\frac{n-1}{n} \tau - \frac{\Theta_-}{2}\right\|_{L^{\frac{2(n-1)^2}{n}}(\partial M, \bR)} c^{\frac{N}{\fraccheloue}}}$.
Finally, if $W$ is the (weak) solution to the vector equation \eqref{eqVector}-\eqref{eqCondVector},
we have, setting $X \equiv W$ in \eqref{eqWeakVector2}:
\begin{align*}
 \frac{1}{2} \int_M |\bL W|^2 \dd\mu^g
  &= \int_{\partial M} \<W, \omega\> \dd\mu^g - \int_M \<W, \gamma\> \dd\mu^g\\
  &\leq \|W\|_{L^{\fraccheloue}(\partial M, TM)} \|\omega\|_{L^{\frac{2(n-1)}{n}}(\partial M, T^*M)}\\
  &\qquad\qquad + \|W\|_{L^N(M, TM)} \|\gamma\|_{L^{\frac{2n}{n+2}}(M, T^*M)}
\end{align*}
As a consequence of the coercivity of $J_0$, defined in \eqref{eqDefJ}
(see Lemma \ref{lmCoercivity2}), and the Sobolev (resp. trace) inequality,
there exists a positive constant $s_V$ (resp. $t_V$) independent of $W$ such
that
\begin{align*}
 \|W\|_{L^N(M, TM)}^2 \leq s_V \int_M |\bL W|^2 \dd\mu^g,\\
 \left(\textrm{resp. } \|W\|_{L^{\fraccheloue}(\partial M, TM)}^2 \leq t_V \int_M |\bL W|^2 \dd\mu^g \right).
\end{align*}
Inserting these inequalities in the previous estimate, we get
\begin{align*}
 &\frac{1}{2} \int_M |\bL W|^2 \dd\mu^g\\
 &\qquad \leq \left(\sqrt{t_V} \|\omega\|_{L^{\frac{2(n-1)}{n}}(\partial M, T^*M)} + \sqrt{s_V} \|\gamma\|_{L^{\frac{2n}{n+2}}(M, T^*M)}\right) \sqrt{\int_M |\bL W|^2 \dd\mu^g}.
\end{align*}
As a consequence,
\begin{equation}\label{eqEstimateR}
 \int_M |\bL W|^2 \dd\mu^g \leq 4 \left(\sqrt{t_V} \|\omega\|_{L^{\frac{2(n-1)}{n}}(\partial M, T^*M)} + \sqrt{s_V} \|\gamma\|_{L^{\frac{2n}{n+2}}(M, T^*M)}\right)^2.
\end{equation}
So $W \in \Omega_R$ (i.e. $\Omega_R$ is stable by $\Phi$) provided that
\[
 \sqrt{t_V} \|\omega\|_{L^{\frac{2(n-1)}{n}}(\partial M, T^*M)} + \sqrt{s_V} \|\gamma\|_{L^{\frac{2n}{n+2}}(M, T^*M)} \leq \frac{R}{2}.
\]
Using Estimates \eqref{eqEstimateGamma} and \eqref{eqEstimateOmega}, we conclude
that $\Omega_R$ is stable as soon as
\begin{equation}\label{eqCondR}
 \zeta_1 (R^2 + \varepsilon^2)^{\frac{N}{N+2}} + \zeta_2 \mu \leq R,
\end{equation}
for some positive constants $\zeta_1$ and $\zeta_2$ that can be explicited from the calculations
above. We now show that, if $\varepsilon$ and $\mu$ are small enough, we can find a
value for $R$ such that Inequality \eqref{eqCondR} is fulfilled.

Note that, if we set $\varepsilon = \mu = 0$, the function
\[
 f_{\varepsilon, \mu}(R) = \zeta_1 (R^2 + \varepsilon^2)^{\frac{N}{N+2}} + \zeta_2 \mu - R
\]
is minimal for $\displaystyle R = R_0 = \left(\frac{N+2}{2N} \zeta_1^{-1}\right)^{\frac{N+2}{N-2}}$.
Seting $R = R_0$ in \eqref{eqCondR}, and rearranging, we get that $\Omega_{R_0}$
is stable provided that
\begin{equation}\label{eqCondExplicit}
 \frac{N+2}{2N} \left[\left(1 + \frac{\varepsilon^2}{R_0^2}\right)^{\frac{N}{N+2}}-1\right] + \frac{\zeta_2 \mu}{R_0} \leq \frac{N-2}{2N}.
\end{equation}
This inequality is trivially satisfied if $\varepsilon$ and $\mu$ are chosen small
enough.
\end{proof}

As a consequence of Propositions \ref{propPhi} and \ref{propInvariant}, the
Schauder fixed point theorem guarantees the existence of a fixed point $W$ to the
map $\Phi$ provided that $\varepsilon$ and $\mu$ are small enough. We will denote by $\phi$ the corresponding (weak) solution to the Lichnerowicz equation.

Our last step consists in making several comings and goings between the regularity theory of the Lichnerowicz equation and the one of the vector equation:\\

$\bullet$ \textsc{Round 1}: From Proposition \ref{propImprovedRegL2}, we know that $\phi^N \in W^{1, r_0}(M, \bR) \cap L^{\fraccheloue}(M, \bR)$. As we assumed $\tau \in W^{1, n}(M, \bR)$, we get from H\"older's inequality that
\[
\gamma \definedas \frac{n-1}{n} \phi^N \dd\tau \in L^{r_0}(M, \bR).
\]
Similarly, we have
\[
\omega \definedas \left(\frac{n-1}{n}\tau 
-\frac{\Theta_-}{2}\right)\phi^{N}\nu^{\flat} \in W^{1, r_0}(M, \bR).
\]
This is because we have $\dsp{f \definedas \frac{n-1}{n}\tau -\frac{\Theta_-}{2} \in W^{1, 2}(M, \bR) \cap L^\infty(M, \bR)}$ so
\[
\dd(f\phi^N) = f \dd\phi^N + \phi^N df
\]
As $f \in L^\infty(M, \bR))$ and $\dd\phi^N \in L^{r_0}(M, \bR)$, $f \dd\phi^N \in L^{r_0}(M, \bR)$. Also, as $\phi^N \in L^{\fraccheloue}(M, \bR)$ and $df \in L^{2p}(M, \bR)$, we have $\phi^N df \in L^{r_0}(M, \bR)$ as
\[
\frac{1}{\dsp{\frac{N}{2}+1}} + \frac{1}{2p} < \frac{1}{\dsp{\frac{N}{2}+1}} + \frac{1}{n}
= \frac{1}{r_0}.
\]
We conclude from Proposition \ref{propVector} that $W \in W^{2, r_0}(M, \bR)$. As a consequence, $A \definedas |\sigma + \bL W| \in L^{2 q_0}(M, \bR)$ with $\dsp{q_0 = \frac{n-1}{n-2}}$. It follows from simple calculations that $\dsp{q_0 > \frac{n^2}{n^2 - n + 2}}$ so, from Lemma \ref{lmLich2}, we conclude that $\phi^N \in W^{1, 2}(M, \bR)$.\\

Note that, if $n = 3$, we have that $q_0 = 2$ so Proposition \ref{propLichStrong} ensures that $\phi$ is actually a strong solution belonging to $W^{2, 2}(M, \bR)$ or $W^{2, p}(M, \bR)$ if $p < 2$. Standard elliptic regularity then implies that we actually have $(\phi, W) \in W^{2, p}(M, \bR) \times W^{2, p}(M, TM)$. This proves Theorem \ref{thmMain} for $n = 3$.\\

$\bullet$ \textsc{Round 2}: Assume now that $n \geqslant 4$. From Round 1, with the use of Hölder inequality, we conclude that $\gamma \in L^2(M, TM)$ and $\omega \in W^{1, 2}(M, TM)$. By Proposition \ref{propVector}, we have that $W \in W^{2, 2}(M, TM)$. Thus, $A := |\sigma + \Lb_g W| \in L^{2 q_0}(M, \Rb)$ with $\dsp{q_0 = \frac{N}{2} < p}$.

Proposition \ref{propLichStrong} gives that $\varphi^N \in L^{r_0}(M, \Rb)$ with
$\dsp{r_0 = \frac{8n(n-1)q}{7n^2 - 2n - (6n+4)q}}$ if $n > 4$ or $r_0 = N - \varepsilon = 4-\varepsilon$, with $\varepsilon > 0$ as small as we want, if $n = 4$ (reducing the value of $q_0$ so we do not attain the critical value $\dsp{q = \frac{n}{2}}$).

We then obtain
$$
\Lb W(\nu, \cdot)
= \left(\frac{n-1}{n} \tau - \frac{\Theta_-}{2}\right) \varphi^N \nu^\flat + \xi
\in W^{1, r_0}(M, \Rb) 
$$
Formally, this estimate comes from the fact that $W^{1, r_0}(M, \Rb)$ is a $W^{1, 2p}(M, \Rb)$-module for the usual multiplication of functions if $r_0 < 2p$. This property is a direct generalization of \cite[Theorem 4.39]{AdamsFournier} where it is proven that $W^{1, 2p}(M, \Rb)$ is a Banach algebra.
Likewise, using Sobolev embedding, we get that $\varphi^N \dd\tau \in L^r(M, T^*M)$. So, using Proposition \ref{propVector}, we get $W \in W^{2, r_0}(M, TM)$, which gives us $A \in L^{2q_1}(M, \Rb)$ with $\dsp{\frac{1}{2q_1} = \frac{1}{r_0} - \frac{1}{n}}$. We thus set a sequence $(q_i)$ by recurrence 
$$
\frac{1}{2 q_{i+1}} = \frac{1}{r_i} - \frac{1}{n},
$$
with
$$
r_i = \frac{8n(n-1)q_i}{7n^2 - 2n - (6n+4)q_i}
$$
(here, we suppose implicitly that $q_i < \dsp{\frac{n}{2}}$, and wish to obtain, after a certain amount of iterations $q_i > \dsp{\frac{n}{2}}$ to get to Round 3). We then get
$$
\frac{1}{q_{i+1}} = - \frac{7n-2}{2n(n-1)} + \frac{7n-2}{4(n-1)} \frac{1}{q_i}.
$$
This is a sequence for the $\dsp{\frac{1}{q_i}}$ of common ratio $\dsp{\frac{7n-2}{4(n-1)} > 1}$, we can rewrite
$$
\frac{1}{q_{i+1}} - \ell = \frac{7n-2}{4(n-1)} \left(\frac{1}{q_i} - \ell\right)
$$
with $\dsp{\ell = \frac{14n-4}{3n^2+2n}}$. The $q_i$ sequence will increase (and exceed $\dsp{\frac{n}{2}}$ after a finite number of iterations) if
$$
\frac{1}{q_0} - \ell < 0.
$$
However, straightforward calculations that this last equality is satisfied if, and only if, $n < 6$. If $n \geqslant 6$, successive iterations do not grant us additional regularity. We will now suppose that $3 \leqslant n \leqslant 5$.\\

$\bullet$ \textsc{Round 3}: We finally have $A \in L^{2 q}(M,  \Rb)$ for a certain $q > \dsp{\frac{n}{2}}$. We can suppose also that $q < n$. Proposition \ref{propLichStrong} gives $\varphi^N \in W^{1, r}(M, \Rb)$ with $\dsp{r = \frac{n q}{n-q} > n}$. We then obtain, using Proposition \ref{propVector} that $W \in W^{2, s}(M, TM)$ with $s = \min\{r, p\}$. In particular, $\Lb W \in L^\infty(M, \Rb)$ and so $A \in L^{2p}(M, \Rb)$, which proves that $\varphi \in W^{2, p}(M, \Rb)$, then $W \in W^{2, p}(M, TM)$ which concludes the proof of our main theorem \ref{thmMain}.

\providecommand{\bysame}{\leavevmode\hbox to3em{\hrulefill}\thinspace}
\providecommand{\MR}{\relax\ifhmode\unskip\space\fi MR }
\providecommand{\MRhref}[2]{%
    \href{http://www.ams.org/mathscinet-getitem?mr=#1}{#2}
}
\providecommand{\href}[2]{#2}


\end{document}